\documentclass[11pt,reqno]{article}

\usepackage{amssymb,amscd,amsmath,amsthm}
\numberwithin{equation}{section}
\usepackage{graphicx}
\newtheorem{theorem}{Theorem}[section]
\newtheorem{proposition}[theorem]{Proposition}
\newtheorem{lemma}[theorem]{Lemma}
\newtheorem{corollary}[theorem]{Corollary}
\newtheorem{definition}[theorem]{Definition}
\newtheorem{remark}[theorem]{Remark}

\def\cb{{\mathcal B}}

\def\ce{{\mathcal E}}

\def\cw{{\mathcal W}}

\def\bc{{\mathbb C}}
\def\bh{{\mathbb H}}

\def\bn{{\mathbb N}}

\def\br{{\mathbb R}}

\def\bz{{\mathbb Z}}

\def\frak{\mathfrak}

\def\ga{{\frak A}}

\def\a{\alpha}
\def\b{\beta}

\def\tr{{\rm Tr}}
\def\L{\Lambda}
\def\G{\Gamma}
\def\ce{\mathcal E}

\def\ffi{\varphi}

\def\<{\langle}
\def\>{\rangle}

\def\1{\mathbf{1}}

\def\cw{\cal W}

\def\cal{\mathcal}

\def\s{\sigma}
\def\bh{\mathbf{h}}
\def\bs{\mathbf{s}}
\def\bg{\mathbf{g}}

\def\id{{\bf 1}\!\!{\rm I}}

\begin{document}

\begin{center}
{\Large {\bf On Quantum Markov Chains on Cayley tree II:\\
 Phase transitions for the associated chain with $XY$-model on \\ the Cayley tree of order three}}\\[1cm]
\end{center}

\begin{center}
{\large {\sc Luigi Accardi}}\\[2mm]
\textit{Centro Interdisciplinare Vito Volterra\\
II Universit\`{a} di Roma ``Tor Vergata''\\
Via Columbia 2, 00133 Roma, Italy} \\
E-email: {\tt accardi@volterra.uniroma2.it}
\end{center}

\begin{center}
{\large {\sc Farrukh Mukhamedov}}\\[2mm]
\textit{ Department of Computational \& Theoretical Sciences,\\
Faculty of Science, International Islamic University Malaysia,\\
P.O. Box, 141, 25710, Kuantan, Pahang, Malaysia}\\
E-mail: {\tt far75m@yandex.ru, \ farrukh\_m@iiu.edu.my}
\end{center}

\begin{center}
{\large{\sc Mansoor Saburov}}\\[2mm]
\textit{Department of Computational \& Theoretical Sciences,\\
Faculty of Science, International Islamic University Malaysia,\\
P.O. Box, 141, 25710, Kuantan, Pahang, Malaysia}\\
E-mail: {\tt msaburov@gmail.com}
\end{center}

\begin{abstract}
In the present paper we study forward  Quantum Markov Chains (QMC)
defined on a Cayley tree. Using the tree structure of graphs, we
give a construction of quantum Markov chains on a Cayley tree. By
means of such constructions we prove the existence of a phase
transition for the $XY$-model on a Cayley tree of order three in QMC
scheme. By the phase transition we mean the existence of two now
quasi equivalent QMC for the given family of interaction operators
$\{K_{<x,y>}\}$.

\end{abstract}

\section{Introduction }\label{intr}

One of the basic open problems in quantum probability is the
construction of a theory of quantum Markov fields, that is quantum
process with  multi-dimensional index set. This program concerns the
generalization of the theory of Markov fields (see
\cite{D},\cite{Geor}) to noncommutative setting, naturally arising
in quantum statistical mechanics and quantum filed theory.

The quantum analogues of Markov chains were first constructed in
\cite{[Ac74f]}, where the notion of quantum Markov chain on infinite
tensor product algebras was introduced. Nowadays, quantum Markov
chains have become a standard computational tool in solid state
physics, and several natural applications have emerged in quantum
statistical mechanics and quantum information theory. The reader is
referred to \cite{fannes2,G,ILW,Kum,Mat,OP} and the references cited
therein, for recent developments of the theory and the applications.

First attempts to construct a quantum analogue of classical Markov
fields have been done in \cite{[Liebs99]}, \cite{[AcFi01a]},
\cite{[AcFi01b]},\cite{AcLi}. In these papers the notion of {\it
quantum Markov state}, introduced in \cite{[AcFr80]}, extended to
fields as a sub--class of the {\it quantum Markov chains} introduced
in \cite{[Ac74f]}. In \cite{[AcFiMu07]} it has been proposed a
definition of quantum Markov states and chains, which extend a
proposed one in \cite{Oh05}, and includes all the presently known
examples. Note that in the mentioned papers quantum Markov fields
were considered over multidimensional integer lattice $\bz^d$. This
lattice has so-called amenability property. On the other hand, there
do not exist analytical solutions (for example, critical
temperature) on such lattice. But investigations of phase
transitions of spin models on hierarchical lattices showed that
there are exact calculations of various physical quantities (see for
example, \cite{Bax,Per}). Such studies on the hierarchical lattices
begun with the development of the Migdal-Kadanoff renormalization
group method where the lattices emerged as approximants of the
ordinary crystal ones.  On the other hand, the study of exactly
solved models deserves some general interest in statistical
mechanics \cite{Bax}. Therefore, it is natural to investigate
quantum Markov fields over hierarchical lattices.  For example, a
Cayley tree is the simplest hierarchical lattice with non-amenable
graph structure. This means that the ratio of the number of boundary
sites to the number of interior sites of the Cayley tree  tends to a
nonzero constant in the thermodynamic limit of a large system, i.e.
the ratio $W_n/V_n$ (see section 2 for the definitions) tends to
$(k-1)/k$ as $ n\to\infty$, where $k$ is the order of the tree.
Nevertheless, the Cayley tree is not a realistic lattice, however,
its amazing topology makes the exact calculation of various
quantities possible. First attempts to investigate quantum Markov
chains over such trees was done in \cite{aklt}, such studies were
related to investigate thermodynamic limit of valence-bond-solid
models on a Cayley tree \cite{fannes}. The mentioned considerations
naturally suggest the study of the following problem: the extension
to fields of the notion of generalized Markov chain. In \cite{AOM}
we have introduced a hierarchy of notions of Markovianity for states
on discrete infinite tensor products of $C^*$--algebras and for each
of these notions we constructed some explicit examples. We showed
that the construction of \cite{[AcFr80]} can be generalized to
trees. It is worth to note that, in a different context and for
quite different purposes, the special role of trees was already
emphasized in \cite{[Liebs99]}. Note that in \cite{fannes} finitely
correlated states are constructed as ground states of VBS-model on a
Cayley tree. Such shift invariant $d$-Markov chains can be
considered as an extension of $C^*$-finitely correlated states
defined in \cite{fannes2} to the Cayley trees. Note that a
noncommutative extension of classical Markov fields, associated with
Ising and Potts models on a Cayley tree, were investigated in
\cite{Mukh04,MR04}. In the classical case, Markov fields on trees
are also considered in \cite{[Pr]}-\cite{[Za85]}.

If a tree is not one-dimensional lattice, then it is expected (from
a physical point of view) the existence of a phase transition for
quantum Markov chains constructed over such a tree. In \cite{AMSa}
we have provided a construction of forward QMC, such states are
different from backward QMC. In that construction, a QMC is defined
as a weak limit of finite volume states with boundary conditions.
Such a QMC depends on the boundary conditions. For by means of the
provided construction we proved uniqueness QMC, associated with
$XY$-model on a Cayley tree of order two.

Our goal, in this paper, is to establish the existence of a phase
transition that $XY$-model on the Cayley tree of order three. Note
that phase transitions in a quantum setting play an important role
to understand quantum spin systems (see for example
\cite{BCS},\cite{FS}). In this paper, using the construction defined
in \cite{AMSa} we shall prove the existence of a phase transition
for the $XY$-model on a Cayley tree of order three in QMC scheme. By
the phase transition we means the existence of two distinct QMC for
the given family of interaction operators $\{K_{<x,y>}\}$. Hence,
results of the present paper will totaly differ from \cite{AMSa},
and show by the increasing the dimension of the tree we are getting
the phase transition. We have to stress here that the constructed
QMC associated with $XY$-model, is different from thermal states of
that model, since such states correspond to $\exp(-\b
\sum_{<x,y>}H_{<x,y>})$, which is different from a product of
$\exp(-\b H_{<x,y>})$. Roughly speaking, if we consider the usual
Hamiltonian system $H(\s)=-\b\sum_{<x,y>}h_{<x,y>}(\s)$, then its
Gibbs measure is defined by the fraction
\begin{equation}\label{mu-G1}
\mu(\s)=\frac{e^{-H(\s)}}{\sum_{\s}e^{-H(\s)}}.
\end{equation}
Such a measure can be viewed in another way as well. Namely,
\begin{equation}\label{mu-G2}
\mu(\s)=\frac{\prod_{<x,y>}e^{\b h_{<x,y>}(\s)}}{\sum\limits_{\s}\prod_{<x,y>}e^{\b h_{<x,y>}(\s)}}.
\end{equation}
A usual quantum mechanical definition of the quantum Gibbs states
based on equation \eqref{mu-G1}. In this paper, we use an
alternative way to define the quantum Gibbs states based on
\eqref{mu-G2}. Note that whether or not the resulting states have a
physical interest is a question that cannot be solved on a purely
mathematical ground.

\section{Preliminaries}\label{dfqmf}

Let $\Gamma^k = (L,E)$ be a semi-infinite Cayley tree of order
$k\geq 1$ with the root $x^0$ (i.e. each vertex of $\Gamma^k$ has
exactly $k+1$ edges, except for the root $x^0$, which has $k$
edges). Here $L$ is the set of vertices and $E$ is the set of edges.
The vertices $x$ and $y$ are called {\it nearest neighbors} and they
are denoted by $l=<x,y>$ if there exists an edge connecting them. A
collection of the pairs $<x,x_1>,\dots,<x_{d-1},y>$ is called a {\it
path} from the point $x$ to the point $y$. The distance $d(x,y),
x,y\in V$, on the Cayley tree, is the length of the shortest path
from $x$ to $y$.

Recall a coordinate structure in $\G^k$:  every vertex $x$ (except
for $x^0$) of $\G^k$ has coordinates $(i_1,\dots,i_n)$, here
$i_m\in\{1,\dots,k\}$, $1\leq m\leq n$ and for the vertex $x^0$ we
put $(0)$.  Namely, the symbol $(0)$ constitutes level 0, and the
sites $(i_1,\dots,i_n)$ form level $n$ ( i.e. $d(x^0,x)=n$) of the
lattice (see Fig. 1).

\begin{figure}
\begin{center}
\includegraphics[width=10.07cm]{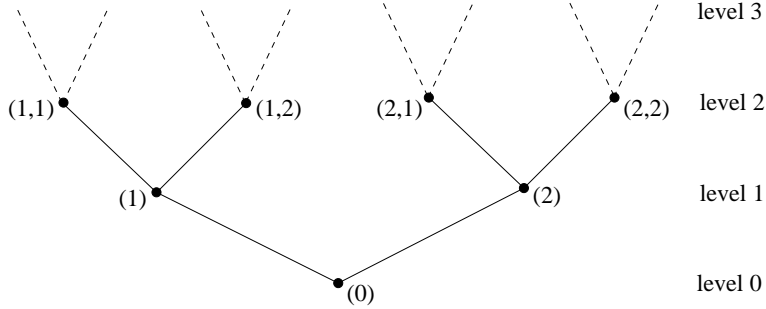}
\end{center}
\caption{The first levels of $\G^2$} \label{fig1}
\end{figure}

Let us set
\[
W_n = \{ x\in L \, : \, d(x,x_0) = n\} , \qquad \Lambda_n =
\bigcup_{k=0}^n W_k, \qquad  \L_{[n,m]}=\bigcup_{k=n}^mW_k, \ (n<m)
\]
\[
E_n = \big\{ <x,y> \in E \, : \, x,y \in \Lambda_n\big\}, \qquad
\Lambda_n^c = \bigcup_{k=n}^\infty W_k
\]
For $x\in \G^k_+$, $x=(i_1,\dots,i_n)$ denote
$$ S(x)=\{(x,i):\ 1\leq
i\leq k\},
$$
here $(x,i)$ means that $(i_1,\dots,i_n,i)$. This set is called a
set of {\it direct successors} of $x$.

The algebra of observables $\cb_x$ for any single site $x\in L$ will
be taken as the algebra $M_d$ of the complex $d\times d$ matrices.
The algebra of observables localized in the finite volume $\L\subset
L$ is then given by $\cb_\L=\bigotimes\limits_{x\in\L}\cb_x$. As
usual if $\L^1\subset\L^2\subset L$, then $\cb_{\L^1}$ is identified
as a subalgebra of $\cb_{\L^2}$ by tensoring with unit matrices on
the sites $x\in\L^2\setminus\L^1$. Note that, in the sequel, by
$\cb_{\L,+}$ we denote the positive part of $\cb_\L$. The full
algebra $\cb_L$ of the tree is obtained in the usual manner by an
inductive limit
$$
\cb_L=\overline{\bigcup\limits_{\L_n}\cb_{\L_n}}.
$$

In what follows, by ${\cal S}({\cal B}_\L)$ we
will denote the set of all states defined on the algebra ${\cal
B}_\L$.

Consider a triplet ${\cal C} \subset {\cal B} \subset {\cal A}$ of
unital $C^*$-algebras. Recall that a {\it quasi-conditional
expectation} with respect to the given triplet is a completely
positive (CP) identity preserving linear map $\ce \,:\, {\cal A} \to
{\cal B}$ such that $ \ce(ca) = c \ce(a)$, for all $a\in {\cal A},\, c \in {\cal C}$.

A state $\varphi$ on ${\cal B}_L$ is called a {\it forward quantum
$d$-Markov chain (QMC)}, associated to $\{\L_n\}$, on ${\cal B}_L$
if for each $\Lambda_n$, there exist a quasi-conditional expectation
$\ce_{\Lambda_n^c}$ with respect to the triplet ${\cal
B}_{{\Lambda}_{n+1}^c}\subseteq {\cal B}_{\Lambda_n^c}\subseteq{\cal
B}_{\Lambda_{n-1}^c} $and a state $\hat\varphi_{\Lambda_n^c}\in{\cal
S}({\cal B}_{\Lambda_n^c}) $ such that for any $n\in {\mathbb N}$
one has
\begin{equation}\label{eq4.1re}
\hat\varphi_{\Lambda_n^c}| {\cal
B}_{\Lambda_{n+1}\backslash\Lambda_n} =
\hat\varphi_{\Lambda_{n+1}^c}\circ \ce_{\Lambda_{n+1}^c}| {\cal
B}_{\Lambda_{n+1}\backslash\Lambda_n}
\end{equation}
and
\begin{equation}\label{dfgqmf}
\varphi = \lim_{n\to\infty} \hat\varphi_{\Lambda_n^c}\circ
\ce_{\Lambda_n^c}\circ \ce_{\Lambda_{n-1}^c} \circ \cdots \circ
\ce_{\Lambda_1^c}
\end{equation}
in the weak-* topology.

Note that \eqref{eq4.1re} is an analogue of the DRL equation from
classical statistical mechanics \cite{D, Geor}, and QMC state is
thus the counterpart of the infinite-volume Gibbs measure.

\begin{remark} We point out that in \cite{AOM} a forward QMC was called a generalized quantum Markov state, and the
existence of the limit \eqref{dfgqmf} under the condition \eqref{eq4.1re} was proved there as well.
\end{remark}

%%%%%%%%%%%%%%%%%%%%%%%%%%%%%%%%%%%%%%%%%%%%%%%%%%%%%%%%%%%% section 5

\section{Construction of QMC on the Cayley tree}\label{dfcayley}

In this section, we recall a construction of forward
quantum $d$-Markov chain (see \cite{AMSa}).

Let us rewrite the elements of $W_n$ in the following order, i.e.
\begin{eqnarray*}
\overrightarrow{W_n}:=\left(x^{(1)}_{W_n},x^{(2)}_{W_n},\cdots,x^{(|W_n|)}_{W_n}\right),\quad \overleftarrow{W_n}:=\left(x^{(|W_n|)}_{W_n},x^{(|W_n|-1)}_{W_n},\cdots, x^{(1)}_{W_n}\right).
\end{eqnarray*}
Note that $|W_n|=k^n$. Vertices $x^{(1)}_{W_n},x^{(2)}_{W_n},\cdots,x^{(|W_n|)}_{W_n}$ of $W_n$ can be
represented in terms of the coordinate system as follows
\begin{eqnarray*}
&&x^{(1)}_{W_n}=(1,1,\cdots,1,1), \quad x^{(2)}_{W_n}=(1,1,\cdots,1,2), \ \ \cdots \quad x^{(k)}_{W_n}=(1,1,\cdots,1,k,),\\
&&x^{(k+1)}_{W_n}=(1,1,\cdots,2,1), \quad x^{(2)}_{W_n}=(1,1,\cdots,2,2), \ \ \cdots \quad x^{(2k)}_{W_n}=(1,1,\cdots,2,k),
\end{eqnarray*}
\[\vdots\]
\begin{eqnarray*}
&&x^{(|W_n|-k+1)}_{W_n}=(k,k,,\cdots,k,1), \ x^{(|W_n|-k+2)}_{W_n}=(k,k,\cdots,k,2),\ \ \cdots  x^{|W_n|}_{W_n}=(k,k,\cdots,k,k).
\end{eqnarray*}

Analogously, for a given vertex $x,$ we shall use the following notation for
the set of direct successors of $x$:
\begin{eqnarray*}
\overrightarrow{S(x)}:=\left((x,1),(x,2),\cdots (x,k)\right),\quad
\overleftarrow{S(x)}:=\left((x,k),(x,k-1),\cdots (x,1)\right).
\end{eqnarray*}
In what follows, for the sake of simplicity, we will use notation
$i\in \overrightarrow{S(x)}$ (resp. $i\in \overleftarrow{S(x)}$
instead of $(x,i)\in \overrightarrow{S(x)}$ (resp. $(x,i)\in
\overleftarrow{S(x)}$).

Assume that for each edge $<x,y>\in E$ of the tree an operator
$K_{<x,y>}\in {\cal B}_{\{x,y\}}$ is assigned. We would like to
define a state on $\cb_{\L_n}$ with boundary conditions $w_{0}\in
{\cal B}_{(0),+}$ and $\bh=\{h_x\in {\cal B}_{x,+}\}_{x\in L}$.

Let us denote
\begin{eqnarray}
K_{[m-1,m]}&:=&\prod_{x\in
\overrightarrow{W}_{m-1}}\prod_{y\in \overrightarrow{S(x)}}K_{<x,y>},\\
\bh^{1/2}_n&:=&\prod_{x\in \overrightarrow{W}_n}h_x^{1/2}, \quad \quad \bh_n:=\bh^{1/2}_n(\bh^{1/2}_n)^{*},\\
K_n&:=&w_0^{1/2}K_{[0,1]}K_{[1,2]}\cdots K_{[n-1,n]}\bh^{1/2}_n,\\
{\cw}_{n]}&:=&K_nK_n^{*},
\end{eqnarray}
It is clear that ${\cw}_{n]}$ is positive.

In what follows, by $\tr_{\L}:\cb_L\to\cb_{\L}$ we mean normalized
partial trace (i.e. $\tr_{\L}(\id_{L})=\id_{\L}$, here
$\id_{\L}=\bigotimes\limits_{y\in \L}\id$), for any
$\Lambda\subseteq_{\text{fin}}L$. For the sake of shortness we put
$\tr_{n]} := \tr_{\Lambda_n}$.

Let us define a positive functional $\ffi^{(n,f)}_{w_0,\bh}$ on
$\cb_{\Lambda_n}$ by
\begin{eqnarray}\label{ffi-ff}
\ffi^{(n,f)}_{w_0,\bh}(a)=\tr(\cw_{n+1]}(a\otimes\id_{W_{n+1}})),
\end{eqnarray}
for every $a\in \cb_{\Lambda_n}$. Note that here, $\tr$ is a
normalized trace on ${\cal B}_L$ (i.e. $\tr(\id_L)=1$).

To get an infinite-volume state $\ffi^{(f)}$ on $\cb_L$  such
that $\ffi^{(f)}\lceil_{\cb_{\L_n}}=\ffi^{(n,f)}_{w_0,\bh}$, we
need to impose some constrains to the boundary conditions
$\big\{w_0,\bh\big\}$ so that the functionals
$\{\ffi^{(n,f)}_{w_0,\bh}\}$ satisfy the compatibility condition,
i.e.
\begin{eqnarray}\label{compatibility}
\ffi^{(n+1,f)}_{w_0,\bh}\lceil_{\cb_{\L_n}}=\ffi^{(n,f)}_{w_0,\bh}.
\end{eqnarray}

\begin{theorem}[\cite{AMSa}]\label{compa} Assume that $K_{<x,y>}$ is self-adjoint for every $<x,y>\in E$. Let the boundary conditions $w_{0}\in {\cal
B}_{(0),+}$ and ${\bh}=\{h_x\in {\cal B}_{x,+}\}_{x\in L}$ satisfy
the following conditions:
\begin{eqnarray}\label{eq1}
&& \tr ( w_0 h_0 ) =1 \\
\label{eq2}
&&\tr_{x]}\left[\prod_{y\in \overrightarrow{S(x)}}K_{<x,y>}\prod_{y\in \overrightarrow{S(x)}}h^{(y)}\prod_{y\in \overleftarrow{S(x)}}K_{<x,y>}\right]=h^{(x)}
\ \ \textrm{for every} \ \  x\in L.
\end{eqnarray}
Then the functionals $\{\ffi^{(n,f)}_{w_0,\bh}\}$ satisfy the
compatibility condition \eqref{compatibility}. Moreover, there is a
unique forward quantum $d$-Markov chain  $\ffi^{(b)}_{w_0,{\bh}}$ on $\cb_L$ such that
$\ffi^{(f)}_{w_0,{\bh}}=w-\lim_{n\to\infty}\ffi^{(n,f)}_{w_0,\bh}$.
\end{theorem}

From direct calculation we can derive the following

\begin{proposition}\label{state^nwithW_n}
If \eqref{eq1} and \eqref{eq2} are satisfied then one has $\ffi^{(n,f)}_{w_0,\bh}(a)=\tr(\cw_{n]}(a))$ for any $a\in \cb_{\Lambda_n}$.
\end{proposition}

Our goal in this paper is to establish the existence of phase
transition for the given family $\{K_{<x,y>}\}$ of operators.
Heuristically, the phase transition means the existence of two
distinct QMC for the given $\{K_{<x,y>}\}$. Let us provide a more
exact definition.

\begin{definition}
We say that there exists a phase transition for a family of
operators $\{K_{<x,y>}\}$ if \eqref{eq1}, \eqref{eq2} have at least
two $(u_0,\{h_x\}_{x\in L})$ and $(v_0,\{s_x\}_{x\in L})$ solutions
such that the corresponding quantum $d$-Markov chains
$\ffi_{u_0,\bh}$ and $\ffi_{v_0,\bs}$ are not quasi equivalent.
Otherwise, we say there is no phase transition.
\end{definition}

\begin{remark} In the classical case, i.e. the interaction operators commute with each other and belong to commutative part of $\cb_L$, the provided definition coincides with the known definition of the phase transition for models with nearest-neighbor interactions on the tree (see for example \cite{Bax,Geor,[Pr]}).

\end{remark}

%%%%%%%%%%%%%%%%%%%%%%%%%%%%%%%%%%%%%%%%%%55     exam1

\section{QMC associated with XY-model and a main result}\label{exam1}

In this section, we define the model and  shall formulate the main
results of the paper. In what follows we consider a semi-infinite
Cayley tree $\G^3=(L,E)$ of order 3. Our starting $C^{*}$-algebra is
the same $\cb_L$ but with $\cb_{x}=M_{2}(\bc)$ for $x\in L$. By
$\s_x^{(u)},\s_y^{(u)},\s_z^{(u)}$ we denote the Pauli spin
operators at site $u\in L$. Here
\begin{equation}\label{pauli}
\s_x^{(u)}= \left(
          \begin{array}{cc}
            0 & 1 \\
            1 & 0 \\
          \end{array}
        \right), \quad
\s_y^{(u)}= \left(
          \begin{array}{cc}
            0 & -i \\
            i & 0 \\
          \end{array}
        \right), \quad
\s_z^{(u)}= \left(
          \begin{array}{cc}
            1 & 0 \\
            0 & -1 \\
          \end{array}
        \right).
\end{equation}

For every edge $<u,v>\in E$ put
\begin{equation}\label{1Kxy1}
K_{<u,v>}=\exp\{\b H_{<u,v>}\}, \ \ \b>0,
\end{equation}
where
\begin{equation}\label{1Hxy1}
H_{<u,v>}=\frac{1}{2}\big(\s_{x}^{(u)}\s_{x}^{(v)}+\s_{y}^{(u)}\s_{y}^{(v)}\big).
\end{equation}

Such kind of Hamiltonian is called {\it quantum XY-model} per edge $<x,y>$.

Now taking into account the following equalities
\begin{eqnarray*}\label{1Hxy2}
&&H_{<u,v>}^{2m}=H_{<u,v>}^2=\frac{1}{2}\big(\id-\s_{z}^{(u)}\s_{z}^{(v)}\big),\
\ \
%\label{1Hxy3}
H_{<u,v>}^{2m-1}=H_{<u,v>}, \ \ \ m\in\bn,
\end{eqnarray*}
one finds
\begin{eqnarray}\label{K<u,v>}
K_{<u,v>}=\id+\sinh\beta H_{<u,v>}+(\cosh\beta-1)H^2_{<u,v>}.
\end{eqnarray}

The main result of the present paper concerns the existence of the
phase transition for the model \eqref{1Kxy1}. Namely, we have

\begin{theorem}\label{main} Let $\{K_{<x,y>}\}$ be given by \eqref{1Kxy1} on the Cayley tree of order three. Then there are two positive numbers $\b_*$ and $\b^*$ such that
\begin{enumerate}
\item[(i)] if $\b\in(0,\b_*]\cup [\b^*,\infty)$, then  there is a unique forward
quantum $d$-Markov chain associated with \eqref{1Kxy1};
\item[(ii)] if $\b\in(\b_*,\b^*)$, then  there is a phase transition for a given model,
i.e. there are two distinct forward quantum $d$-Markov chains.
\end{enumerate}
\end{theorem}

The rest of the paper will be devote to the proof the this theorem.
To do it, we shall use a dynamical system approach, which is
associated with the equations \eqref{eq1},\eqref{eq2}.

\section{A dynamical system related to \eqref{eq1},\eqref{eq2}}

In this section we shall reduce equations \eqref{eq1},\eqref{eq2} to
some dynamical system. Our goal is to describe all solutions
$\bh=\{h_x\}$ and $w_0$ of those equations.

Furthermore, we shall assume that $h_x=h_y$ for every $x,y\in W_n$,
$n\in\bn$. Hence, we denote $h_x^{(n)}:=h_x$, if $x\in W_n$. Now
from \eqref{1Kxy1},\eqref{1Hxy1} one can see that
$K_{<u,u>}=K^{*}_{<u,v>}$, therefore, equation \eqref{eq2} can be
rewritten as follows
\begin{eqnarray}\label{state}
Tr_x(K_{<x,y>}K_{<x,z>}K_{<x,v>}h^{(n)}_yh^{(n)}_zh^{(n)}_vK_{<x,v>}K_{<x,z>}K_{<x,y>})&=&h_x^{(n-1)},
\end{eqnarray}
for every $x\in L.$

After small calculations equation \eqref{state} reduces to the
following system
\begin{equation}\label{mainsystem}
\left\{
\begin{array}{r}
   \left(\dfrac{a^{(n)}_{11}+a^{(n)}_{22}}{2}\right)^3B_2+a^{(n)}_{12}a^{(n)}_{21}
   \left(\dfrac{a^{(n)}_{11}+a^{(n)}_{22}}{2}\right)A_2 = a^{(n-1)}_{11} \\
   a^{(n)}_{12}\left(\left(\dfrac{a^{(n)}_{11}+a^{(n)}_{22}}{2}\right)^2B_1
   +a^{(n)}_{12}a^{(n)}_{21}A_1\right)= a^{(n-1)}_{12} \\
   a^{(n)}_{21}\left(\left(\dfrac{a^{(n)}_{11}+a^{(n)}_{22}}{2}\right)^2B_1
   +a^{(n)}_{12}a^{(n)}_{21}A_1\right)= a^{(n-1)}_{21} \\
   \left(\dfrac{a^{(n)}_{11}+a^{(n)}_{22}}{2}\right)^3B_2+a^{(n)}_{12}a^{(n)}_{21}
   \left(\dfrac{a^{(n)}_{11}+a^{(n)}_{22}}{2}\right)A_2 = a^{(n-1)}_{22}
\end{array}
\right.
\end{equation}
where
\begin{eqnarray}\label{A1B1}
A_{1}=\sinh^3\beta\cosh\beta,\quad B_{1}=\sinh\beta\cosh^{2}\beta(1+\cosh\beta+\cosh^2\beta),\\
A_{2}=\sinh^2\beta\cosh^2\beta(1+2\cosh\beta),\quad
B_2=\cosh^6\beta.\label{A2B2}
\end{eqnarray}
Here
\begin{equation*}
h_{x}^{(n-1)}=\left(
          \begin{array}{cc}
            a^{(n-1)}_{11} & a^{(n-1)}_{12} \\
            a^{(n-1)}_{21} & a^{(n-1)}_{22} \\
          \end{array}
        \right), \quad\quad
h_{y}^{(n)}=h_{z}^{(n)}=h_{v}^{(n)}=\left(
          \begin{array}{cc}
            a^{(n)}_{11} & a^{(n)}_{12} \\
            a^{(n)}_{21} & a^{(n)}_{22} \\
          \end{array}
        \right).
\end{equation*}
From  \eqref{mainsystem} we immediately get that
$a^{(n)}_{11}=a^{(n)}_{22}$ for all $n\in \bn$.

Self-adjointness of $h_x^{(n)}$ (i.e.
$\overline{a^{(n)}_{12}}=a^{(n)}_{21},$ for any $n\in \bn$) and the
representation $a_{12}^{(n)}=|a_{12}^{(n)}|\exp(i\varphi_{n})$
allows us to reduce the system \eqref{mainsystem} to
\begin{equation}\label{equationtohxnwithrealandphi}
\left\{
\begin{array}{r}
B_2(a^{(n)}_{11})^3+A_2a^{(n)}_{11}|a^{(n)}_{12}|^2= a^{(n-1)}_{11}\\
|a^{(n)}_{12}|\left(B_1(a^{(n)}_{11})^2+A_1|a^{(n)}_{12}|^2\right)=
|a^{(n-1)}_{12}|\\
\varphi_{n}=\varphi_{n-1}
\end{array}
\right.
\end{equation}

From \eqref{equationtohxnwithrealandphi} it follows that $\varphi_{n}=\varphi_{0}$, whenever $n\in\bn.$ Therefore, we shall study the following system

\begin{equation}\label{equationtohxnwithreal}
\left\{
\begin{array}{r}
B_2(a^{(n)}_{11})^3+A_2a^{(n)}_{11}|a^{(n)}_{12}|^2
= a^{(n-1)}_{11}\\
|a^{(n)}_{12}|\left(B_1(a^{(n)}_{11})^2+A_1|a^{(n)}_{12}|^2\right)=
|a^{(n-1)}_{12}|
\end{array}
\right.
\end{equation}

\begin{remark}\label{positivityofhxn} Note that according to
the positivity of $h_x^{(n)}$ and $a_{11}^{(n)}=a_{22}^{(n)}$ we conclude that
$a_{11}^{(n)}>|a_{12}^{(n)}|$   for all $n\in\bn.$
\end{remark}

Now we are going to investigate the derived system
\eqref{equationtohxnwithreal}. To do this, let us define
a mapping $f:(x,y)\in\br^2_{+}\to(x{'},y{'})\in\br^2_{+}$ by
\begin{equation}\label{dynsystem}
\left\{
\begin{array}{r}
B_2(x^{'})^3+A_2x^{'}(y^{'})^2= x\\
B_1(x^{'})^2y^{'}+A_1(y^{'})^3=y
\end{array}
\right.
\end{equation}

Furthermore, due to Remark \ref{positivityofhxn}, we restrict the
dynamical system \eqref{dynsystem} to the following domain
$$\Delta=\{(x,y)\in \br^2_{+}: x > y\}.$$

Denote
\begin{eqnarray}
&&P_9(t)=t^9-t^8-t^7-t^6+2t^4+2t^3-t-1,\label{P9(t)}\\
&&D:=\frac{A_2-A_1}{B_1-B_2}.\label{D}\\
&&E:=\frac{1}{A_2+DB_2},\label{E}
\end{eqnarray}

Further, we will need the following auxiliary facts.

\begin{lemma}\label{inequalities}
Let $A_1,B_1,A_2,B_2,D$ be numbers defined by \eqref{A1B1},
\eqref{A2B2}, \eqref{D} and $P_9(t)$ be the polynomial given by
\eqref{P9(t)}, where \ $\beta>0.$ Then the following statements hold
true:
\begin{itemize}
  \item [(i)] The polynomial $P_9(t)$ has only three positive roots 1, $t_{*},$ and $t^{*}$ such that $1.05<t_{*}<1.1$ and  $1.5<t_{*}<1.6.$ Moreover, if $t\in (1,t_{*})\cup (t^{*},\infty)$ then $P_9(t)>0$ and $t\in (t_{*},t^{*})$ then $P_9(t)<0.$ Denote by $\beta_{*}=\cosh^{-1}t_{*}$ and $\beta^{*}=\cosh^{-1}t^{*};$
  \item [(ii)] For any $\beta\in (0,\infty)$ we have $A_1< A_{2};$
  \item [(iii)] If $\beta\in(0,\beta_{*}]\cup[\beta^{*},\infty)$ then $B_1\le B_2$ and  If $\beta\in(\beta_{*},\beta^{*})$ then $B_1> B_2;$
  \item [(iv)] For any $\beta\in (0,\infty)$ we have $A_1+B_1 < A_{2}+B_2;$
  \item [(v)] If $\beta\in(\beta_{*},\beta^{*})$ then $D>1$ and $E>0;$
  \item [(vi)] For any $\beta\in (0,\infty)$ we have $A_1A_2< B_1B_2$ and $A_1B_2< A_2B_1;$
  \item [(vii)] If $\beta\in(\beta_{*},\beta^{*})$ then $A_2B_1<A_1A_2+3A_1B_2+B_1B_2$  and $2A_1A_2+3A_1B_2<A_2B_1;$
  \item [(viii)] For any $\beta\in (0,\infty)$ we have  $0<\sinh\beta(1+\cosh\beta)<\cosh^3\beta.$
\end{itemize}

\end{lemma}

The proof is provided in the Appendix.

\section{Fixed points and asymptotical behavior of $f$. Existence of forward QMC}

In this section we shall find fixed points of \eqref{dynsystem} and
prove the absence of periodic points. Moreover, we investigate an
asymptotical behavior of \eqref{dynsystem}. Note that every fixed
point of \eqref{dynsystem} defines (see Theorem \ref{compa}) a
forward QMC. Hence, the existence of the fixed points implies the
existence of forward QMC.

Let us first find all of the fixed points of the system.

\begin{theorem}\label{fixed-p}
Let $f$ be a dynamical system given by \eqref{dynsystem}. Then the
following assertions hold true:
\begin{enumerate}
\item[(i)] If $\beta\in(0,\beta_{*}]\cup[\beta^{*},\infty)$ then there is a unique fixed point $\big(\frac{1}{\cosh^3\beta},0\big)$ in the
domain $\Delta$;

\item[(ii)] If $\beta\in(\beta_{*},\beta^{*})$ then there are two fixed points in the domain $\Delta,$ which are $\big(\frac{1}{\cosh^3\beta},0\big)$ and $(\sqrt{DE},\sqrt{E}).$
\end{enumerate}
\end{theorem}

\begin{proof}
Assume that $(x,y)$ is a fixed point, i.e.
\begin{equation}\label{1fix}
\left\{
\begin{array}{r}
B_2x^3+A_2xy^2
= x\\
B_1x^2y+A_1y^3=y
\end{array}
\right.
\end{equation}

Consider two different cases with respect to $y$.

{\sc Case (a).} Let $y=0.$ Then one finds that either $x=0$ or
$x=\frac{1}{\cosh^3\beta}.$ But, only the point
$(\frac{1}{\cosh^3\beta},0)$ belongs to the domain $\Delta.$\\

{\sc Case (b).} Now suppose $y>0.$ Since $x>y>0$ one
finds
\begin{eqnarray*}
\left\{
\begin{array}{r}
B_2x^2+A_2y^2= 1\\
B_1x^2+A_1y^2=1,
\end{array}
\right.
\end{eqnarray*}
hence, due to \eqref{A1B1} and \eqref{A2B2} we obtain
\begin{eqnarray*}
(B_1-B_2)x^2=
(A_2-A_1)y^2.
\end{eqnarray*}

According to  Lemma \ref{inequalities} $(ii),(iii),(v)$  we infer
that if $\beta\in(0,\beta_{*}]\cup[\beta^{*},\infty)$ then $B_1\le
B_2,$ $A_1< A_2$, and if $\beta\in(\beta_{*},\beta^{*})$ then
$B_1>B_2,$ $A_1< A_2$, and which imply
$$\frac{x^2}{y^2}=\frac{A_2-A_1}{B_1-B_2}=D>1.$$
Therefore, if $\beta\in(0,\beta_{*}]\cup[\beta^{*},\infty)$ then the dynamical system \eqref{dynsystem} has a unique fixed point $(\frac{1}{\cosh^3\beta},0)$. If  $\beta\in(\beta_{*},\beta^{*})$ then the dynamical system \eqref{dynsystem} has two fixed points
$(\frac{1}{\cosh^3\beta},0)$ and $(\sqrt{DE},\sqrt{E}).$
\end{proof}

To investigate  an asymptotical behavior of the dynamical system on  $\Delta$ we need some auxiliary facts.

Let $g_{\beta}:[0,1]\rightarrow {\mathbb{R}}_{+}$ be  a function  given by
\begin{eqnarray}\label{gbetat}
g_{\beta}(t)=\frac{A_1t^3+B_1t}{A_2t^2+B_2},
\end{eqnarray}
where $\beta\in (0,\infty).$

\begin{proposition}\label{gincreasingfunction} Let $g_{\beta}:[0,1]\rightarrow {\mathbb{R}}_{+}$ be the function given by \eqref{gbetat} and $\beta\in(\beta_{*},\beta^{*}).$ Then the following assertions hold true:
\begin{enumerate}
\item[(i)] $g_{\beta}$ is an increasing function on $[0,1]$;\\
\item[(ii)] If $t\in [0,\frac{1}{\sqrt{D}}]$, then $g_{\beta}(t)\ge t$. If $t\in [\frac{1}{\sqrt{D}},1]$ then $g_{\beta}(t)\le t$;\\
\item[(iii)]  If $0\le g_{\beta}(t)\le\frac{1}{\sqrt{D}}$ then $0\le t\le\frac{1}{\sqrt{D}}$ and if $\frac{1}{\sqrt{D}}\le g_{\beta}(t)\le1$ then $\frac{1}{\sqrt{D}}\le t\le1.$
\end{enumerate}
\end{proposition}
\begin{proof} Let us prove (i). We know that
\begin{eqnarray*}
g_{\beta}'(t)=\frac{A_1A_2t^4+(3A_1B_2-A_2B_1)t^2+B_1B_2}{(A_2t^2+B_2)^2}.
\end{eqnarray*}
Let us denote
\[\hat{g}(t)=A_1A_2t^4+(3A_1B_2-A_2B_1)t^2+B_1B_2.\]
It is enough to show that $\hat{g}(t)>0,$ for any $t\in [0,1].$ To do so, we will show that $\min\limits_{t\in[0,1]}\hat{g}(t)>0.$
It follows from Lemma \ref{inequalities} $(vii)$ that $\hat{g}(0)=B_1B_2>0$ and $\hat{g}(1)>0.$ It is clear that
$$\hat{g}'(t)=2t(2A_1A_2t^2-(A_2B_1-3A_1B_2))$$
Since $A_2B_1-3A_1B_2-2A_1A_2>0$ (see Lemma \ref{inequalities} $(vii)$) one has
\[t^2=\frac{A_2B_1-3A_1B_2}{2A_1A_2}>1.\] So, $\min\limits_{t\in[0,1]}\hat{g}(t)>0$, and hence $\hat{g}(t)>0$ for any $t\in [0,1].$
Therefore, $g_{\beta}'(t)>0$ for any $t\in [0,1],$ and this proves the assertion.

(ii). One can see that
\begin{eqnarray}\label{gbetaminust}
g_\beta(t)-t=-\frac{(A_2-A_1)t(t^2-\frac{1}{D})}{A_2t^2+B_2}
\end{eqnarray}
Therefore, we find that if $t\in [0,\frac{1}{\sqrt{D}}]$ then $g_{\beta}(t)\ge t$, and  if $t\in [\frac{1}{\sqrt{D}},1]$ then $g_{\beta}(t)\le t.$

(iii). It follows from \eqref{gbetaminust} that the function $g_\beta(t)$ has two fixed points $t=0$ and $t=\frac{1}{\sqrt{D}}.$  Let  $0\le g_{\beta}(t)\le\frac{1}{\sqrt{D}}$, and suppose that $t>\frac{1}{\sqrt{D}}.$
Due to (i) and  $t=\frac{1}{\sqrt{D}}$ is fixed point, we obtain  $g_{\beta}(t)>\frac{1}{\sqrt{D}}$, which is impossible. Similarly, one can show that $\frac{1}{\sqrt{D}}\le g_{\beta}(t)\le1$ implies $\frac{1}{\sqrt{D}}\le t\le1.$
\end{proof}

Let us start to study the asymptotical behavior of the dynamical system $f:\Delta\rightarrow {\mathbb{R}}_{+}$ given by \eqref{dynsystem}

\begin{theorem}\label{absentofperiodicpointoff}
The dynamical system $f:\Delta\rightarrow {\mathbb{R}}^2_{+}$, given
by \eqref{dynsystem} (with $\beta\in(0,\infty)$), does not have any
$k$ ($k\geq 2$) periodic points in $\Delta$.
\end{theorem}
\begin{proof}
Assume that the dynamical system $f$  has a periodic
point $(x^{(0)},y^{(0)})$ with a period of $k$ in $\Delta,$ where $k\geq 2.$
This means that there are points
$$(x^{(0)},y^{(0)}),(x^{(1)},y^{(1)}),\dots,(x^{(k-1)},y^{(k-1)})\in \Delta,$$
such that they satisfy the following equalities
\begin{equation}\label{conforperiodic}
\left\{
\begin{array}{r}
B_2(x^{(i)})^3+A_2x^{(i)}(y^{(i)})^2= x^{(i-1)}\\
B_1(x^{(i)})^2y^{(i)}+A_1(y^{(i)})^3=y^{(i-1)}
\end{array}
\right.
\end{equation}
where $i=\overline{1,k},$ i.e.
$f\left(x^{(i-1)},y^{(i-1)}\right)=\left(x^{(i)},y^{(i)}\right),$ with  $x^{(k)}=x^{(0)},$
$y^{(k)}=y^{(0)}.$

Now again consider two different cases with respect to $y^{(0)}$.\\

{\sc Case (a).} Let $y^{(0)}>0.$ Then $x^{(i)},y^{(i)}$ should
be positive for all $i=\overline{1,k}.$ Let us look for different cases with respect to $\beta.$

Assume that $\beta\in (0,\beta_{*}]\cup[\beta^{*},\infty).$ We then have
\begin{eqnarray*}
\frac{x^{(i-1)}}{y^{(i-1)}} &=&
\frac{B_2}{B_1}\cdot\frac{x^{(i)}}{y^{(i)}}+\frac{(A_2B_1-A_1B_2)x^{(i)}y^{(i)}}{B_1\left(B_1(x^{(i)})^2+A_1(y^{(i)})^2\right)}
\end{eqnarray*}
where $i=\overline{1,k}.$

Due to Lemma \ref{inequalities} $(vi)$ and $x^{(i)},y^{(i)}>0$ for all $i=\overline{1,k},$ one finds
\begin{eqnarray}\label{ineqperoid}
\frac{x^{(i-1)}}{y^{(i-1)}} >
\frac{B_2}{B_1}\cdot\frac{x^{(i)}}{y^{(i)}},
\end{eqnarray}
for all  $i=\overline{1,k}.$

Iterating  \eqref{ineqperoid} we get
\begin{eqnarray*}
\frac{x^{(0)}}{y^{(0)}}>
\left(\frac{B_2}{B_1}\right)^{k}\cdot\frac{x^{(0)}}{y^{(0)}}.
\end{eqnarray*}
But, the last inequality is impossible, since Lemma
\ref{inequalities} $(iii)$ implies
$$\frac{B_2}{B_1}\ge1.$$
Hence, in this case, the dynamical system \eqref{dynsystem} does not
have any periodic points with $k\geq 2.$

Let $\beta\in (\beta_{*},\beta^{*})$,  then one finds
\begin{eqnarray*}
\frac{y^{(i-1)}}{x^{(i-1)}} =
g_{\beta}\left(\frac{y^{(i)}}{x^{(i)}}\right),\ \ \
i=\overline{1,k}.
\end{eqnarray*}
This means that $\frac{y^{(0)}}{x^{(0)}}$ is a $k$
periodic point for the function $g_\beta(t).$ But this contradictions to
Proposition \ref{gincreasingfunction} (i), since
the function $g_{\beta}(t)$ is increasing, and it does not have any periodic point on the segment $[0,1].$\\

{\sc Case (b).} Now suppose that $y^{(0)}=0.$ Since $k\geq 2$ we have
$x^{(0)}\neq \frac{1}{\cosh^3\beta}.$ So, from
\eqref{conforperiodic} one finds that $y^{(i)}=0$ for all
$i=\overline{1,k}.$ Then again \eqref{conforperiodic} implies that
$$(x^{(i)})^3\cosh^6\beta=x^{(i-1)}, \ \ \ \forall
i=\overline{1,k},$$ which means
$$x^{(i)}=\frac{1}{\cosh^2\beta}\sqrt[3]{x^{(i-1)}}, \ \ \ \forall
i=\overline{1,k}.$$ Hence, we have
$$x^{(0)}=\frac{1}{\cosh^{3}\beta}\sqrt[\leftroot{-2}\uproot{3}3^{k+1}]{x^{(0)}\cosh^3\beta}.$$
This yields either $x^{(0)}=0$ or $x^{(0)}= \frac{1}{\cosh^3\beta},$
which is a contradiction.
\end{proof}

\begin{theorem}\label{trajectorywhenbetain0beta*infty}
Let  $f:\Delta\to{\mathbb{R}}^2_{+}$ be the dynamical system given by
\eqref{dynsystem} and $\beta\in (0,\beta_{*}]\cup[\beta^{*},\infty).$ Then the following assertions hold true:
\begin{enumerate}
\item[(i)] if $y^{(0)}>0$ then the trajectory
$\{(x^{(n)},y^{(n)})\}_{n=0}^{\infty}$ of $f$ starting from the
point $(x^{(0)},y^{(0)})$ is finite.

\item[(ii)] if $y^{(0)}=0$ then the trajectory
$\{(x^{(n)},y^{(n)})\}_{n=0}^{\infty}$  starting from the point
$(x^{(0)},y^{(0)})$ has the following form
\begin{equation*}
\left\{
\begin{array}{l}
x^{(n)} = \cfrac{\sqrt[3^n]{x^{(0)}\cosh^3\beta}}{\cosh^3\beta}\\
y^{(n)} = 0,
\end{array}
\right.
\end{equation*}
and it converges to the fixed point $(\frac{1}{\cosh^3\beta},0).$
\end{enumerate}
\end{theorem}

\begin{proof} (i) Let $y^{(0)}>0$ and  suppose that the trajectory
$\{(x^{(n)},y^{(n)})\}_{n=0}^{\infty}$ of the dynamical system
starting from the point $(x^{(0)},y^{(0)})$ is infinite. This means
that the points $(x^{(n)},y^{(n)})$ are well defined and belong to the domain $\Delta$ for all
$n\in\bn.$ Since $y^{(0)}>0$ we have $y^{(n)}>0$ for all  $n\in\bn.$  Then, it follows from \eqref{dynsystem} that
\begin{eqnarray}\label{xn-1yn-1xnyn}
\frac{x^{(n-1)}}{y^{(n-1)}} =
\frac{B_2}{B_1}\cdot\frac{x^{(n)}}{y^{(n)}}+\frac{(A_2B_1-A_1B_2)\cfrac{x^{(n)}}{y^{(n)}}}{B_1^2\left(\cfrac{x^{(n)}}{y^{(n)}}\right)^2+A_1B_1}
\quad \textrm{for all} \ \ n\in\bn.
\end{eqnarray}
 It yields that
\begin{eqnarray*}
\frac{x^{(n-1)}}{y^{(n-1)}} >
\frac{B_2}{B_1}\cdot\frac{x^{(n)}}{y^{(n)}},
\end{eqnarray*}
and
\begin{eqnarray}\label{x0y0andxnyn}
\frac{x^{(0)}}{y^{(0)}} >
\left(\frac{B_2}{B_1}\right)^n\cdot\frac{x^{(n)}}{y^{(n)}}, \quad \textrm{for all} \ \ n\in\bn.
\end{eqnarray}
  It follows from \eqref{x0y0andxnyn} and Lemma \ref{inequalities} $(iii)$ that
\begin{eqnarray}\label{xnynx0y0}
\frac{x^{(n)}}{y^{(n)}}<
\left(\frac{B_1}{B_2}\right)^n\cdot\frac{x^{(0)}}{y^{(0)}}\le\frac{x^{(0)}}{y^{(0)}},
\end{eqnarray}
for all $n\in\bn.$ Using \eqref{xn-1yn-1xnyn} and \eqref{xnynx0y0} one  gets
\begin{eqnarray*}
\frac{x^{(n-1)}}{y^{(n-1)}} >
\left(\frac{B_2}{B_1}+\frac{(A_2B_1-A_1B_2)(y^{(0)})^2}{B_1^2(x^{(0)})^2+A_1B_1(y^{(0)})^2}\right)\cdot\frac{x^{(n)}}{y^{(n)}},
\end{eqnarray*}
and
\begin{eqnarray*}
\frac{x^{(n)}}{y^{(n)}}<
\left(\frac{B_2}{B_1}+\frac{(A_2B_1-A_1B_2)(y^{(0)})^2}{B_1^2(x^{(0)})^2+A_1B_1(y^{(0)})^2}\right)^{-n}\cdot\frac{x^{(0)}}{y^{(0)}}.
\end{eqnarray*}
We know that if $\beta\in (0,\beta_{*}]\cup[\beta^{*},\infty)$ then due to Lemma \ref{inequalities} $(iii)$ one finds
\[\frac{B_2}{B_1}+\frac{(A_2B_1-A_1B_2)(y^{(0)})^2}{B_1^2(x^{(0)})^2+A_1B_1(y^{(0)})^2}>\frac{B_2}{B_1}\ge1.\]
Therefore, we conclude that, for all $\beta\in (0,\beta_{*}]\cup[\beta^{*},\infty)$
\begin{eqnarray*}
\frac{x^{(n)}}{y^{(n)}}\to 0
\end{eqnarray*}
as $n\to \infty.$

On the other hand, due to  $(x^{(n)},y^{(n)})\in \Delta,$ we have
\begin{eqnarray*}\label{conforxnyn}
\frac{x^{(n)}}{y^{(n)}}\ge 1,
\end{eqnarray*}
for all $n\in\bn.$ This contradiction shows that the trajectory
$\{(x^{(n)},y^{(n)})\}_{n=0}^{\infty}$ must be finite.\\

(ii) Now let $y^{(0)}=0$, then \eqref{dynsystem} implies
$y^{(n)}=0$ for all $n\in\bn.$ Hence, from \eqref{dynsystem} one
finds
$$x^{(n)}\cosh^3\beta=\sqrt[3]{x^{(n-1)}\cosh^3\beta}.$$
So, iterating the last equality we obtain
$$x^{(n)}\cosh^3\beta=\sqrt[3^n]{x^{(0)}\cosh^3\beta},$$
which yields the desired equality and the trajectory $\{(x^{(n)},0)\}_{n=0}^{\infty}$ converges to the fixed point $(\frac{1}{\cosh^3\beta},0).$
\end{proof}

\begin{theorem}\label{trajectorywhenbetainbeta_*beta*}
Let  $f:\Delta\to{\mathbb{R}}^2_{+}$ be the dynamical system given by
\eqref{dynsystem} and $\beta\in (\beta_{*},\beta^{*}).$ Then the following assertions hold true:
\begin{enumerate}
\item[(i)] There are  two invariant lines  $l_1=\{(x,y)\in\Delta: y=0\}$ and $l_2=\{(x,y)\in\Delta: y=\frac{x}{\sqrt{D}}\}$ w.r.t. $f$;
 \item [(ii)] if an initial point $(x^{(0)},y^{(0)})$ belongs to the invariant
 line $l_k,$ then its trajectory $\{(x^{(n)},y^{(n)})\}_{n=0}^{\infty}$
 converges to the  fixed point belonging to the line $l_k,$ where $k=\overline{1,2};$
  \item [(iii)] if an initial point $(x^{(0)},y^{(0)})$ satisfies the following condition
      \[\frac{y^{(0)}}{x^{(0)}}\in \left(0,\frac{1}{\sqrt{D}}\right),\] then its trajectory $\{(x^{(n)},y^{(n)})\}_{n=0}^{\infty}$ converges to the fixed point $(\frac{1}{\cosh^3\beta},0)$ which belongs to $l_1;$
  \item [(iv)] if an initial point $(x^{(0)},y^{(0)})$ satisfies the following condition
      \[\frac{y^{(0)}}{x^{(0)}}\in \left(\frac{1}{\sqrt{D}},1\right),\] then  its trajectory $\{(x^{(n)},y^{(n)})\}_{n=0}^{\infty}$  is finite.
\end{enumerate}
\end{theorem}

\begin{proof} (i). It follows from \eqref{dynsystem} that
if $y=0$ then $y'=0,$ which means $l_1$ is an invariant line. Let
$\frac{y}{x}=\frac{1}{\sqrt{D}}.$ Again from \eqref{dynsystem} it
follows that
$\frac{1}{\sqrt{D}}=\frac{y}{x}=g_{\beta}(\frac{y'}{x'}).$ Since
$g_{\beta}(t)$ is an increasing function on segment $[0,1]$ and
$t=\frac{1}{\sqrt{D}}$ is its fixed point, we then get
$\frac{y'}{x'}=\frac{1}{\sqrt{D}},$ which yields that $l_2$ is an
invariant line for $f$.

(ii).  Let us consider a case when an initial point $(x^{(0)},y^{(0)})$ belongs
to $l_k$. Let $(x_k,y_k)$ be the fixed point of $f$ belonging to $l_k$ ($k=\overline{1,2}$). It follows from \eqref{dynsystem} that
\begin{eqnarray}\label{y0x0gbetaynxn}
\frac{y_k}{x_k}=\frac{y^{(0)}}{x^{(0)}}=g^{(n)}_\beta\left(\frac{y^{(n)}}{x^{(n)}}\right).
\end{eqnarray}
for all $n\in \bn.$ Since $g_\beta(t)$ is increasing and $t=\frac{y_k}{x_k}$ is its fixed point, we have
\begin{eqnarray}\label{y0x0ynxn}
\frac{y_k}{x_k}=\frac{y^{(n)}}{x^{(n)}},
\end{eqnarray}
for all $n\in \bn.$ We know that $\frac{y_1}{x_1}=0$ and $\frac{y_2}{x_2}=\frac{1}{\sqrt{D}}.$

In the case when $\frac{y_1}{x_1}=0$, one gets $$\left(x^{(n)},y^{(n)}\right)=\left(\frac{\sqrt[3^n]{x^{(0)}\cosh^3\beta}}{\cosh^3\beta},0\right),$$ hence the trajectory converges to the fixed point $(x_1,y_1)=(\frac{1}{\cosh^3\beta},0)$. Clearly, it belongs to $l_1.$

In the case when $\frac{y_2}{x_2}=\frac{1}{\sqrt{D}},$ we have
$$\left(x^{(n)},y^{(n)}\right)=\left(\sqrt{DE}\sqrt[3^n]{\frac{x^{(0)}}{\sqrt{DE}}},\sqrt{E}\sqrt[3^n]{\frac{y^{(0)}}{\sqrt{E}}}\right),$$ and
the trajectory converges to the fixed point $(x_2,y_2)=(\sqrt{DE},\sqrt{E})$ which belongs to the line $l_2.$\\

(iii). Assume that an initial point $(x^{(0)},y^{(0)})$
satisfies
\begin{equation}\label{0-D}
\frac{y^{(0)}}{x^{(0)}}\in \left(0,\frac{1}{\sqrt{D}}\right).
\end{equation}
It then follows from \eqref{dynsystem} that
\begin{eqnarray*}
\frac{y^{(n-1)}}{x^{(n-1)}}=g_\beta\left(\frac{y^{(n)}}{x^{(n)}}\right),
\end{eqnarray*}
for all $n\in \bn.$ Since \eqref{0-D} and due to Proposition \eqref{gincreasingfunction} (ii), we conclude that
\[\frac{y^{(n)}}{x^{(n)}}\in \left(0,\frac{1}{\sqrt{D}}\right),\]
for all $n\in \bn.$ According to Proposition \ref{gincreasingfunction}(iii) we get
\[\frac{y^{(0)}}{x^{(0)}}>\frac{y^{(1)}}{x^{(1)}}>\cdots>\frac{y^{(n)}}{x^{(n)}}>\cdots,\]
and the sequence $$c_n:=\cfrac{y^{(n)}}{x^{(n)}}$$ converges to $0.$

Let us denote $$b_n:=\frac{1}{B_2+c_{n}A_2}.$$ From \eqref{dynsystem}, one can easily get
\[x^{(n)}=\sqrt[3]{b_{n}\sqrt[3]{b_{n-1}\sqrt[3]{\cdots{\sqrt[3]{b_1x^{(0)}}}}}}\]
and \[\lim\limits_{n\to\infty}x^{(n)}=\lim\limits_{n\to\infty}b_n=\frac{1}{\sqrt{B_2}}=\frac{1}{\cosh^3\beta}.\] Therefore, the trajectory $\{(x^{(n)},y^{(n)})\}_{n=0}^{\infty}$ converges to the fixed point $(\frac{1}{\cosh^3\beta},0)$ which belongs to $l_1.$\\

(iv) Now assume that
\begin{equation}\label{D-1}
\frac{y^{(0)}}{x^{(0)}}\in \left(\frac{1}{\sqrt{D}},1\right).
\end{equation}
We suppose that the trajectory
$\{(x^{(n)},y^{(n)})\}_{n=0}^{\infty}$ is infinite. This means
that the points $(x^{(n)},y^{(n)})$ are well defined and belong to the domain $\Delta$ for all
$n\in\bn.$ Then, it follows from \eqref{dynsystem} that
\begin{eqnarray*}
\frac{y^{(n-1)}}{x^{(n-1)}}=g_\beta\left(\frac{y^{(n)}}{x^{(n)}}\right),
\end{eqnarray*}
for all $n\in \bn.$ Since \eqref{D-1} and  due to Proposition \eqref{gincreasingfunction} (ii), we conclude that
\[\frac{y^{(n)}}{x^{(n)}}\in \left(\frac{1}{\sqrt{D}},1\right),\]
for all $n\in \bn.$ According to Proposition \ref{gincreasingfunction}(iii) one finds
\[\frac{y^{(0)}}{x^{(0)}}<\frac{y^{(1)}}{x^{(1)}}<\cdots<\frac{y^{(n)}}{x^{(n)}}<\cdots.\] Since $(x^{(n)},y^{(n)})\in\Delta$ and the sequence $\cfrac{y^{(n)}}{x^{(n)}}$ is bounded, so it converges to some point $\tilde{t}\in(\frac{1}{\sqrt{D}},1].$ We know that the point  $\tilde{t}$ should be a fixed point of $g_\beta(t)$ on $(\frac{1}{\sqrt{D}},1].$ However, the function $g_\beta(t)$ does not have any fixed points on $(\frac{1}{\sqrt{D}},1].$ Hence, this contradiction shows that the trajectory
$\{(x^{(n)},y^{(n)})\}_{n=0}^{\infty}$ must be finite.
\end{proof}

\section{Uniqueness of QMC}

In this section we prove the first part of the main theorem (see
Theorem \ref{main}), i.e. we show the uniqueness of the forward
quantum $d$-Markov chain in the regime $\b\in(0,\b_*)\cup
[\b^*,\infty)$.

So, assume that $\beta\in(0,\beta_{*}]\cup[\beta^{*},\infty).$
From Theorem \ref{trajectorywhenbetain0beta*infty}, we infer that equations
\eqref{eq1},\eqref{eq2} have a lot of parametrical solutions
$(w_0(\a),\{h_x(\a)\})$ given by
\begin{equation}\label{solutionofmainstate}
w_0(\alpha)=\left(
              \begin{array}{cc}
                \dfrac{1}{\alpha} & 0 \\
                0 & \dfrac{1}{\alpha} \\
              \end{array}
            \right),\quad
 h^{(n)}_x(\alpha)=\left(
                    \begin{array}{cc}
                      \dfrac{\sqrt[3^{n}]{\alpha\cosh^3\beta}}{\cosh^3\beta} & 0 \\
                      0 & \dfrac{\sqrt[3^{n}]{\alpha\cosh^3\beta}}{\cosh^3\beta} \\
                    \end{array}
                 \right),
                 \end{equation}
for every $x\in V$, here $\alpha$ is any positive real number.

The boundary conditions corresponding to the fixed point of
\eqref{dynsystem} are the following ones:
\begin{equation}\label{solutionofmainstatewhenalphafixed}
w_0=\left(
               \begin{array}{cc}
                      {\cosh^3\beta} & 0 \\
                      0 & {\cosh^3\beta} \\
                    \end{array}
                  \right), \quad
                  h^{(n)}_x=\left(
                   \begin{array}{cc}
                \dfrac{1}{\cosh^3\beta} & 0 \\
                0 & \dfrac{1}{\cosh^3\beta} \\
              \end{array}
            \right), \ \ \forall x\in V,
            \end{equation}
which correspond to the value of $\alpha_0=\cfrac{1}{\cosh^3\beta}$
in \eqref{solutionofmainstate}. Therefore, in the sequel we denote such
operators by $w_0\left(\alpha_0\right)$ and
$h_x^{(n)}\left(\alpha_0\right)$, respectively.

Let us consider the states $\ffi^{(n,f)}_{w_0(\a),\bh(\alpha)}$
corresponding to the solutions
$(w_0(\alpha),\{h_x^{(n)}(\alpha)\})$. By definition we have
\begin{eqnarray}\label{uniq}
\ffi^{(n,f)}_{w_0(\a),\bh(\alpha)}(x) &=&
\tr\left(w^{1/2}_{0}(\alpha)\prod_{i=0}^{n-1}K_{[i,i+1]}\prod_{x\in  \overrightarrow{W}_n}h^{(n)}_x(\alpha)
\prod_{i=1}^{n}K_{[n-i,n+1-i]}w^{1/2}_{0}(\alpha)x\right)\nonumber\\
&=&\frac{\left(\sqrt[3^{n+1}]{\alpha\cosh^4\beta}\right)^{3^{n+1}}}{{\alpha}(\cosh^4\beta)^{3^{n+1}}}
\tr\left(\prod_{i=0}^{n-1}K_{[i,i+1]}\prod_{i=1}^{n}K_{[n-i,n+1-i]}x\right)\nonumber\\
&=&\frac{\alpha_0^{3^{n+1}}}{\alpha_0}\tr\left(\prod_{i=0}^{n-1}K_{[i,i+1]}\prod_{i=1}^{n}K_{[n-i,n+1-i]}x\right)\nonumber\\
&=& \tr\left((w^{1/2}_{0}(\alpha_0)\prod_{i=0}^{n-1}K_{[i,i+1]}\prod_{x\in  \overrightarrow{W}_n}h^{(n)}_x(\alpha_0)
\prod_{i=1}^{n}K_{[n-i,n+1-i]}w^{1/2}_{0}(\alpha_0)x\right)\nonumber\\
&=&\ffi^{(n,f)}_{w_0(\a_0),\bh(\alpha_0)}(x),
\end{eqnarray}
for any $\alpha$. Hence, from the definition of quantum $d-$Markov
chain we find that
$\ffi^{(f)}_{w_0(\a),\bh(\alpha)}=\ffi^{(f)}_{w_0(\a_0),\bh(\alpha_0)}$,
which yields that the uniqueness of forward quantum $d$-Markov chain
associated with the model \eqref{1Kxy1}.

Hence, Theorem \ref{main} (i) is proved.

\section{Existence of phase transition}

This section is devoted to the proof of part (ii) of Theorem \ref{main}. We shall prove the existence of the
phase transition in the regime  $\beta\in(\beta_{*},\beta^{*}).$

In this section, for the sake of simplicity of formulas, we will use
the following notations, for the Pauli matrices:
\begin{eqnarray*}
\sigma_0:=\id, \quad \sigma_1:=\sigma_x, \quad \sigma_2:=\sigma_y, \quad \sigma_3:=\sigma_z
\end{eqnarray*}

According to Theorem \ref{fixed-p} in the considered regime there
are two fixed points of the dynamical system  \eqref{dynsystem}.
Then the corresponding solutions of equations
\eqref{eq1},\eqref{eq2} can be written as follows:
$(w_0(\a_0),\{h_x(\a_0)\})$ and $(w_0(\gamma),\{h_x(\gamma)\})$,
where
\begin{eqnarray*}
&&w_0(\a_0)=\frac{1}{\a_0}\sigma_0, \quad h_x(\a_0)=\a_0\sigma^{(x)}_0\\[2mm]
&&w_0(\gamma)=\frac{1}{\gamma_0}\sigma_0, \quad h_x(\gamma)=\gamma_0\sigma^{(x)}_0+\gamma_1\sigma_1^{(x)}
\end{eqnarray*}
here $\a_0=\frac{1}{\cosh^3\beta}$, $\gamma=(\gamma_0,\gamma_1)$ with $\gamma_0=\sqrt{DE},$ $\gamma_1=\sqrt{E}$.

By $\ffi^{(f)}_{w_0(\a_0),{\bh(\gamma)}}$, $\ffi^{(f)}_{w_0(\gamma),{\bh(\gamma)}}$ we denote the corresponding forward quantum $d-$Markov chains. To prove the existence of the phase transition, we need to show that these two states are not quasi-equivalent. To do so, we will need some auxiliary facts and results.\\

Denote
\begin{equation}\label{mainMatrix}
A=\left(
    \begin{array}{cc}
      \cosh^6\beta\gamma_0^2+\sinh^2\beta\cosh^3\beta\gamma_1^2 & \gamma_0\gamma_1\sinh^2\beta\cosh^2\beta(1+\cosh\beta) \\
      \gamma_0\gamma_1\sinh\beta\cosh^2\beta(1+\cosh\beta) & \sinh\beta\cosh^4\beta\gamma_0^2+\sinh^3\beta\cosh\beta\gamma_1^2 \\
    \end{array}
  \right).
\end{equation}

Let us study some properties of the matrix $A.$
One can easily check out that the matrix $A$ given by \eqref{mainMatrix} can be written as follows
\begin{equation}\label{mainformofmainMatrix}
A=\left(
    \begin{array}{cc}
      \dfrac{\cosh\beta(\sinh\beta+\cosh^3\beta)}{\sinh\beta(1+\cosh\beta)^2} & \dfrac{\sqrt{(A_2-A_1)(B_1-B_2)}}{\sinh\beta\cosh^2\beta(1+\cosh\beta)^2} \\
       \dfrac{\sqrt{(A_2-A_1)(B_1-B_2)}}{\sinh^2\beta\cosh^2\beta(1+\cosh\beta)^2} &
        \dfrac{\sinh\beta+\cosh^3\beta}{\cosh\beta(1+\cosh\beta)^2}  \\
    \end{array}
  \right).
\end{equation}
\begin{proposition}\label{A-m}
If $\beta\in(\beta_{*},\beta_{*})$ then the following inequalities  hold true
\begin{itemize}
  \item [(i)] $0<\dfrac{\cosh\beta(\sinh\beta+\cosh^3\beta)}{\sinh\beta(1+\cosh\beta)^2}<1;$
  \item [(ii)] $0< \dfrac{\sinh\beta+\cosh^3\beta}{\cosh\beta(1+\cosh\beta)^2}<1;$
  \item [(iii)] $1<\tr(A)<2;$
  \item [(iv)] $0<\det(A)<1.$
\end{itemize}
\end{proposition}
\begin{proof}
(i). Since $B_2<B_1$ (see Lemma \ref{inequalities} $(iii)$) one can see that
\[0<\dfrac{\cosh\beta(\sinh\beta+\cosh^3\beta)}{\sinh\beta(1+\cosh\beta)^2}
=\frac{\dfrac{B_2}{\cosh^2\beta}+\cosh\beta\sinh\beta}{\dfrac{B_1}{\cosh^2\beta}+\cosh\beta\sinh\beta}<1.\]

(ii). The inequality  $\sinh\beta<\cosh\beta$ implies that
\[0< \dfrac{\sinh\beta+\cosh^3\beta}{\cosh\beta(1+\cosh\beta)^2}=\frac{\sinh\beta+\cosh^3\beta}{\cosh\beta+2\cosh^2\beta+\cosh^3\beta}<1.\]

(iii). One can see that
\begin{eqnarray}\label{trofmainMatrix}
\tr(A)=\dfrac{(\sinh\beta+\cosh^2\beta)(\sinh\beta+\cosh^3\beta)}{\sinh\beta\cosh\beta(1+\cosh\beta)^2}.
\end{eqnarray}
Therefore, from (i), (ii) it immediately  follows that $0<\tr(A)<2.$ Now we are going to show that $\tr(A)>1.$
Indeed, since $\cosh^3\beta>\sinh\beta(1+\cosh\beta)>0$ (see Lemma \ref{inequalities} $(viii)$) and $\cosh\beta>1$ one has
\begin{eqnarray}\label{fortrandDet}
\sinh^2\beta+\cosh^5\beta>\sinh\beta\cosh\beta(1+\cosh\beta)
\end{eqnarray}
Then, due to  \eqref{fortrandDet} we find
\[\tr(A)=\dfrac{\sinh^2\beta+\cosh^5\beta+\sinh\beta\cosh^2\beta(1+\cosh\beta)}
{\sinh\beta\cosh\beta(1+\cosh\beta)+\sinh\beta\cosh^2\beta(1+\cosh\beta)}>1.\]

(iv). Let us evaluate the  determinant $\det(A)$ of the matrix $A$ given by \eqref{mainformofmainMatrix}.
After some algebraic manipulations, one finds
\begin{eqnarray}\label{detofmainMatrix}
\det(A)&=&\frac{\sinh^2\beta+\cosh^5\beta-\sinh\beta\cosh\beta(1+\cosh\beta)}{\sinh\beta\cosh\beta(1+\cosh\beta)^2}.
\end{eqnarray}
Due to \eqref{fortrandDet} one can see that $\det(A)>0.$ We want to show that $\det(A)<1.$
Since $B_2<B_1$ (see Lemma \ref{inequalities} (iii)) and $\sinh\beta<\cosh\beta$ we have
\begin{eqnarray}
\label{B1B2}
&&\cosh^5\beta<\sinh\beta\cosh\beta(1+\cosh\beta+\cosh^2\beta),\\
\label{sinhcosh}&&\sinh^2\beta<\sinh\beta\cosh\beta(1+2\cosh\beta).
\end{eqnarray}
From  inequalities \eqref{B1B2},\eqref{sinhcosh}, one gets
\begin{eqnarray}
\sinh^2\beta+\cosh^5\beta<\sinh\beta\cosh\beta(2+3\cosh\beta+\cosh^2\beta).
\end{eqnarray}
Therefore, we obatain
\begin{eqnarray*}
\det(A)=\frac{\sinh^2\beta+\cosh^5\beta-\sinh\beta\cosh\beta(1+\cosh\beta)}
{\sinh\beta\cosh\beta(2+3\cosh\beta+\cosh^2\beta)-\sinh\beta\cosh\beta(1+\cosh\beta)}<1.
\end{eqnarray*}
This completes the proof.
\end{proof}

The next proposition deals with eigenvalues of the matrix $A$.

\begin{proposition}\label{A-N} Let $A$ be the matrix given by \eqref{mainformofmainMatrix}. Then the following assertions hold true:
\begin{itemize}
  \item [(i)] the numbers $\lambda_1=1$, $\lambda_2=\det(A)$ are eigenvalues of the matrix $A;$
  \item [(ii)] the vectors
  \begin{eqnarray}
  \label{eigenvectorlambda1}(x_1,y_1)&=&\left(\dfrac{\sqrt{(A_2-A_1)(B_1-B_2)}}{\sinh\beta\cosh^2\beta(1+\cosh\beta)^2},
  \dfrac{B_1-B_2}{\sinh\beta\cosh^2\beta(1+\cosh\beta)^2}\right),\\
  \label{eigenvectorlambda2}(x_2,y_2)&=&\left(\dfrac{B_2-B_1}{\sinh\beta\cosh^2\beta(1+\cosh\beta)^2},
  \dfrac{\sqrt{(A_2-A_1)(B_1-B_2)}}{\sinh^2\beta\cosh^2\beta(1+\cosh\beta)^2}\right)
  \end{eqnarray} are eigenvectors of the matrix $A$ corresponding to the eigenvalues $\lambda_1=1$ and $\lambda_2=\det(A),$ respectively;
  \item [(iii)] if $P=\left(
              \begin{array}{cc}
                x_1 & x_2 \\
                y_1 & y_2 \\
              \end{array}
            \right),$ where the vectors $(x_1,y_1)$ and $(x_2,y_2)$ are defined by \eqref{eigenvectorlambda1}, \eqref{eigenvectorlambda2} then \begin{eqnarray}\label{diagonalformofmainMatrix}
            P^{-1}AP=\left(\begin{array}{cc}
                                 \lambda_1 & 0 \\
                                  0 & \lambda_2 \\
                           \end{array}
                     \right);
            \end{eqnarray}
  \item [(iv)] for any $n\in \bn$ one has
  \begin{eqnarray}
  A^{n}=\left(
         \begin{array}{cc}
            \dfrac{x_1^2+\lambda_2^ny_1^2\sinh\beta}{x_1^2+y_1^2\sinh\beta} & \dfrac{x_1y_1\sinh\beta(1-\lambda_2^n)}{x_1^2+y_1^2\sinh\beta} \\
            \dfrac{x_1y_1(1-\lambda_2^n)}{x_1^2+y_1^2\sinh\beta} & \dfrac{\lambda_2^nx_1^2+y_1^2\sinh\beta}{x_1^2+y_1^2\sinh\beta}\\
         \end{array}
    \right).
  \end{eqnarray}
\end{itemize}
\end{proposition}
\begin{proof}
$(i)$ We know that the following equation
\begin{equation*}
\lambda^2-\tr(A)\lambda+\det(A)=0
\end{equation*}
is a characteristic equation of the matrix $A$ given by \eqref{mainformofmainMatrix}. From \eqref{trofmainMatrix} and \eqref{detofmainMatrix} one can easily see that
\begin{eqnarray*}
\tr(A)-\det(A)=\dfrac{\sinh\beta\cosh^2\beta(1+\cosh\beta)+\sinh\beta\cosh\beta(1+\cosh\beta)}
{\sinh\beta\cosh\beta(1+\cosh\beta)^2}=1,
\end{eqnarray*}
this means that $\lambda_1=1$ and $\lambda_2=\det(A)$ are eigenvalues of the matrix $A.$

$(ii)$ The eigenvector $(x_1,y_1)$ of the matrix $A$, corresponding to $\lambda_1=1$ satisfies the following equation
\begin{eqnarray*}
\left(\frac{\cosh\beta(\sinh\beta+\cosh^3\beta)}{\sinh\beta(1+\cosh\beta)^2}-\lambda_1\right)x_1
+\frac{\sqrt{(A_2-A_1)(B_1-B_2)}}{\sinh\beta\cosh^2\beta(1+\cosh\beta)^2}y_1=0.
\end{eqnarray*}
Then, one finds
\begin{eqnarray*}\left\{\begin{array}{l}
x_1=\dfrac{\sqrt{(A_2-A_1)(B_1-B_2)}}{\sinh\beta\cosh^2\beta(1+\cosh\beta)^2}\\
y_1= \lambda_1-\dfrac{\cosh\beta(\sinh\beta+\cosh^3\beta)}{\sinh\beta(1+\cosh\beta)^2}=\dfrac{B_1-B_2}{\sinh\beta\cosh^2\beta(1+\cosh\beta)^2}.
\end{array}
\right.
\end{eqnarray*}
Analogously, one can show that the eigenvector $(x_2,y_2)$ of the matrix $A$, corresponding to $\lambda_2=\det(A)$, is equal to
\begin{eqnarray*}\left\{\begin{array}{l}
x_2=\lambda_2-\dfrac{\sinh\beta+\cosh^3\beta}{\cosh\beta(1+\cosh\beta)^2}=\dfrac{B_2-B_1}{\sinh\beta\cosh^2\beta(1+\cosh\beta)^2}\\
y_2= \dfrac{\sqrt{(A_2-A_1)(B_1-B_2)}}{\sinh^2\beta\cosh^2\beta(1+\cosh\beta)^2}.
\end{array}
\right.
\end{eqnarray*}
It is worth noting that $(x_2,y_2)=\left(-y_1,\dfrac{x_1}{\sinh\beta}\right).$

$(iii)$ Let \[P=\left(
              \begin{array}{cc}
                x_1 & x_2 \\
                y_1 & y_2 \\
              \end{array}
            \right),\]
where the vectors $(x_1,y_1)$ and $(x_2,y_2)$ are defined by \eqref{eigenvectorlambda1}, \eqref{eigenvectorlambda2}. We then get
\begin{eqnarray*}
P^{-1}AP&=&\frac{1}{\det(P)}\left(
                              \begin{array}{cc}
                                y_2 & -x_2 \\
                                -y_1 & x_1 \\
                              \end{array}
                            \right)
\left(
       \begin{array}{cc}
         \lambda_1x_1 & \lambda_2x_2 \\
         \lambda_1y_1 & \lambda_2y_2 \\
       \end{array}
     \right)=\left(
               \begin{array}{cc}
                 \lambda_1 & 0 \\
                 0 & \lambda_2 \\
               \end{array}
             \right),
\end{eqnarray*}
where $\det(P)=\dfrac{x_1^2}{\sinh\beta}+y_1^2>0.$

$(iv)$ From \eqref{diagonalformofmainMatrix} it follows that
\[A=P\left(
       \begin{array}{cc}
         \lambda_1 & 0 \\
         0 & \lambda_2 \\
       \end{array}
     \right)P^{-1}.
\]
Therefore, for any $n\in\bn$ we obtain
\begin{eqnarray*}
A^n&=&P\left(
       \begin{array}{cc}
         \lambda_1^n & 0 \\
         0 & \lambda_2^n \\
       \end{array}
     \right)P^{-1}=\frac{1}{\det(P)}\left(
              \begin{array}{cc}
                x_1 & x_2 \\
                y_1 & y_2 \\
              \end{array}
            \right)\left(
                          \begin{array}{cc}
                             y_2\lambda_1^n & -x_2\lambda_1^n \\
                             -y_1\lambda_2^n & x_1\lambda_2^n \\
                          \end{array}
                     \right)\\
&=&\frac{1}{\det(P)}\left(
                           \begin{array}{cc}
                             x_1y_2\lambda_1^n-x_2y_1\lambda_2^n & x_1x_2(\lambda_2^n-\lambda_1^n) \\
                             y_1y_2(\lambda_1^n-\lambda_2^n) & x_1y_2\lambda_2^n-x_2y_1\lambda_1^n \\
                           \end{array}
                         \right)\\
&=&\left(
         \begin{array}{cc}
            \dfrac{x_1^2+\lambda_2^ny_1^2\sinh\beta}{x_1^2+y_1^2\sinh\beta} & \dfrac{x_1y_1\sinh\beta(1-\lambda_2^n)}{x_1^2+y_1^2\sinh\beta} \\
            \dfrac{x_1y_1(1-\lambda_2^n)}{x_1^2+y_1^2\sinh\beta} & \dfrac{\lambda_2^nx_1^2+y_1^2\sinh\beta}{x_1^2+y_1^2\sinh\beta}\\
         \end{array}
    \right).
\end{eqnarray*}
This completes the proof.
\end{proof}

In what follows, for the sake of simplicity, let us denote
\begin{eqnarray}\label{numbersK_i}
K_0:=\frac{1+\cosh\beta}{2},\quad K_1:=\frac{\sinh\beta}{2}, \quad K_2:=\frac{\sinh\beta}{2}, \quad K_3:=\frac{1+\cosh\beta}{2}
\end{eqnarray}
In these notations, the operator $K_{<u,v>}$ given by \eqref{K<u,v>} can be written as follows
\begin{eqnarray}\label{compactformofK<u,v>}
K_{<u,v>}=\sum\limits_{i=0}^{3}K_{i}\sigma_i^{(u)}\otimes\sigma_i^{(v)}.
\end{eqnarray}
\begin{remark} In the sequel,  we will frequently use the following identities for the numbers $K_i,$ $i=\overline{0,3}$ given by \eqref{numbersK_i}:
\begin{itemize}
  \item [(i)] $K_0^2+K_1^2+K_2^2+K_3^2=\cosh^2\beta;$
  \item [(ii)] $2(K_0K_1-K_2K_3)=\sinh\beta\cosh\beta;$
  \item [(iii)] $2(K_0K_1+K_2K_3)=\sinh\beta;$
  \item [(iv)] $K_0^2+K_1^2-K_2^2-K_3^2=\cosh\beta.$
\end{itemize}
\end{remark}

\begin{proposition}
Let $K_{<u,v>}$ be given by \eqref{compactformofK<u,v>}, $\overrightarrow{S(x)}=(1,2,3),$ and $\bh^{(i)}=h_0^{(i)}\sigma_0^{(i)}+h_1^{(i)}\sigma_1^{(i)},$ where $i\in \overrightarrow{S(x)}.$ Then we have
\begin{eqnarray}
\tr_{x]}\left[\prod_{i\in \overrightarrow{S(x)}}K_{<x,i>}\prod_{i\in \overrightarrow{S(x)}}\bh^{(i)}\prod_{i\in \overleftarrow{S(x)}}K_{<x,i>}\right]=h_0^{(x)}\sigma_0^{(x)}+h_1^{(x)}\sigma_1^{(x)}
\end{eqnarray}
where
\begin{eqnarray}
h_0^{(x)}&=&h_0^{(1)}h_0^{(2)}h_0^{(3)}\cosh^6\beta+h_0^{(1)}h_1^{(2)}h_1^{(3)}\sinh^2\beta\cosh^3\beta\nonumber\\
&&+h_1^{(1)}h_1^{(2)}h_0^{(3)}\sinh^2\beta\cosh^3\beta+h_1^{(1)}h_0^{(2)}h_1^{(3)}\sinh^2\beta\cosh^2\beta,\\[2mm]
h_1^{(x)}&=&h_0^{(1)}h_0^{(2)}h_1^{(3)}\sinh\beta\cosh^2\beta+h_0^{(1)}h_1^{(2)}h_0^{(3)}\sinh\beta\cosh^3\beta\nonumber\\
&&+h_1^{(1)}h_0^{(2)}h_0^{(3)}\sinh\beta\cosh^4\beta+h_1^{(1)}h_1^{(2)}h_1^{(3)}\sinh^3\beta\cosh\beta
\end{eqnarray}
\end{proposition}
\begin{proof}
Let us first evaluate ${\bg}_3^{(x)}:=\tr_{x]}\left[K_{<x,3>}\bh^{(3)}K_{<x,3>}\right].$ From \eqref{compactformofK<u,v>} it follows that
\begin{eqnarray*}
K_{<x,3>}\bh^{(3)}K_{<x,3>}&=&\sum\limits_{i,j=0}^3K_iK_j\sigma_{i}^{(x)}\sigma_{j}^{(x)}\otimes\sigma_i^{(3)}
\left(h_0^{(3)}\sigma_0^{(3)}+h_1^{(3)}\sigma_1^{(3)}\right)\sigma_j^{(3)}\\
&=& h_{0}^{(3)}\sum\limits_{i,j=0}^3K_iK_j\sigma_{i}^{(x)}\sigma_{j}^{(x)}\otimes\sigma_i^{(3)}\sigma_j^{(3)}\\
&& + h_{1}^{(3)}\sum\limits_{i,j=0}^3K_iK_j\sigma_{i}^{(x)}\sigma_{j}^{(x)}\otimes\sigma_i^{(3)}\sigma_1^{(3)}\sigma_j^{(3)}
\end{eqnarray*}
Therefore, one gets
\begin{eqnarray}\label{g3ofTrh3}
{\bg}_3^{(x)}&=&g_0^{(3)}\sigma_{0}^{(x)}+g_1^{(3)}\sigma_{1}^{(x)}
\end{eqnarray}
where
\begin{eqnarray}
\label{g_0^(3)} g_0^{(3)}&=& h_{0}^{(3)}(K_0^2+K_1^2+K_2^2+K_3^2)=h_0^{(3)}\cosh^2\beta\\
\label{g_1^(3)} g_1^{(3)}&=& 2h_1^{(3)}(K_0K_1+K_2K_3)=h_1^{(3)}\sinh\beta.
\end{eqnarray}

Now, evaluate ${\bg}_2^{(x)}:=\tr_{x]}\left[K_{<x,2>}\bh^{(2)}{\bg}_3^{(x)}K_{<x,2>}\right].$ Using \eqref{compactformofK<u,v>} and \eqref{g3ofTrh3} we find
\begin{eqnarray*}
K_{<x,2>}\bh^{(2)}{\bg}_3^{(x)}K_{<x,2>}&=&g_0^{(3)}h_0^{(2)}\sum\limits_{i,j=0}^3K_iK_j\sigma_{i}^{(x)}\sigma_{j}^{(x)}
\otimes\sigma_i^{(2)}\sigma_j^{(2)}\\
&&+g_0^{(3)}h_1^{(2)}\sum\limits_{i,j=0}^3K_iK_j\sigma_{i}^{(x)}\sigma_{j}^{(x)}\otimes\sigma_i^{(2)}\sigma_1^{(2)}\sigma_j^{(2)}\\
&&+g_1^{(3)}h_0^{(2)}\sum\limits_{i,j=0}^3K_iK_j\sigma_{i}^{(x)}\sigma_1^{(x)}\sigma_{j}^{(x)}\otimes\sigma_i^{(2)}\sigma_j^{(2)}\\
&&+g_1^{(3)}h_1^{(2)}\sum\limits_{i,j=0}^3K_iK_j\sigma_{i}^{(x)}\sigma_1^{(x)}\sigma_{j}^{(x)}\otimes\sigma_i^{(2)}\sigma_1^{(2)}\sigma_j^{(2)}.\\
\end{eqnarray*}
Hence, one has
\begin{eqnarray}
{\bg}_2^{(x)}=g_0^{(2)}\sigma_{0}^{(x)}+g_1^{(2)}\sigma_{1}^{(x)}
\end{eqnarray}
where
\begin{eqnarray}
\label{g_0^(2)} g_0^{(2)}&=&g_0^{(3)}h_0^{(2)}(K_0^2+K_1^2+K_2^2+K_3^2)+2g_1^{(3)}h_1^{(2)}(K_0K_1-K_2K_3) \nonumber\\
&=&g_0^{(3)}h_0^{(2)}\cosh^2\beta+g_1^{(3)}h_1^{(2)}\sinh\beta\cosh\beta,\\[2mm]
\label{g_1^(2)} g_1^{(2)}&=&2g_0^{(3)}h_1^{(2)}(K_0K_1+K_2K_3)+g_1^{(3)}h_0^{(2)}(K_0^2+K_1^2-K_2^2-K_3^2)\nonumber\\
&=&g_0^{(3)}h_1^{(2)}\sinh\beta+g_1^{(3)}h_0^{(2)}\cosh\beta.
\end{eqnarray}

Similarly, one can evaluate
\begin{eqnarray}
{\bg}_1^{(x)}:=\tr_{x]}\left[K_{<x,1>}\bh^{(1)}{\bg}_2^{(x)}K_{<x,1>}\right]=g_0^{(1)}\sigma_{0}^{(x)}+g_1^{(1)}\sigma_{1}^{(x)}
\end{eqnarray}
where
\begin{eqnarray}
\label{g_0^(1)} g_0^{(1)}&=&g_0^{(2)}h_0^{(1)}\cosh^2\beta+g_1^{(2)}h_1^{(1)}\sinh\beta\cosh\beta,\\
\label{g_1^(1)} g_1^{(2)}&=&g_0^{(2)}h_1^{(1)}\sinh\beta+g_1^{(2)}h_0^{(1)}\cosh\beta.
\end{eqnarray}
We know that
\[\tr_{x]}\left[\prod_{i\in \overrightarrow{S(x)}}K_{<x,i>}\prod_{i\in \overrightarrow{S(x)}}\bh^{(i)}\prod_{i\in \overleftarrow{S(x)}}K_{<x,i>}\right]={\bg}_1^{(x)},\]
and combining
\eqref{g_0^(3)},\eqref{g_1^(3)},\eqref{g_0^(2)},\eqref{g_1^(2)},\eqref{g_0^(1)},\eqref{g_1^(1)},
we get
\begin{eqnarray*}
g_0^{(1)}&=&h_0^{(1)}h_0^{(2)}h_0^{(3)}\cosh^6\beta+h_0^{(1)}h_1^{(2)}h_1^{(3)}\sinh^2\beta\cosh^3\beta\\
&&+h_1^{(1)}h_1^{(2)}h_0^{(3)}\sinh^2\beta\cosh^3\beta+h_1^{(1)}h_0^{(2)}h_1^{(3)}\sinh^2\beta\cosh^2\beta,\\
g_1^{(2)}&=&h_0^{(1)}h_0^{(2)}h_1^{(3)}\sinh\beta\cosh^2\beta+h_0^{(1)}h_1^{(2)}h_0^{(3)}\sinh\beta\cosh^3\beta\\
&&+h_1^{(1)}h_0^{(2)}h_0^{(3)}\sinh\beta\cosh^4\beta+h_1^{(1)}h_1^{(2)}h_1^{(3)}\sinh^3\beta\cosh\beta
\end{eqnarray*}
This completes the proof.
\end{proof}

\begin{corollary}\label{traceforalpha_0}
Let $K_{<u,v>}$ be  given by \eqref{compactformofK<u,v>}, $\overrightarrow{S(x)}=(1,2,3),$ and
\begin{eqnarray*}
\bh^{(1)}=h_1\sigma_1^{(1)}, \quad
\bh^{(2)}=\alpha_0\sigma_0^{(2)}, \quad \bh^{(3)}=\alpha_0\sigma_0^{(3)}.
\end{eqnarray*}
Then we have
\begin{eqnarray}
\tr_{x]}\left[\prod_{i\in \overrightarrow{S(x)}}K_{<x,i>}\prod_{i\in \overrightarrow{S(x)}}\bh^{(i)}\prod_{i\in \overleftarrow{S(x)}}K_{<x,i>}\right]=\alpha_0^2h_1\sinh\beta\cosh^4\beta\sigma_1^{(x)}
\end{eqnarray}
\end{corollary}
\begin{corollary}\label{traceforgamma_0gamma_1}
Let $K_{<u,v>}$ be  given by \eqref{compactformofK<u,v>}, $\overrightarrow{S(x)}=(1,2,3),$ and
\begin{eqnarray*}
\bh^{(1)}=h_0\sigma_0^{(1)}+h_1\sigma_1^{(1)},\\
\bh^{(2)}=\gamma_0\sigma_0^{(2)}+\gamma_1\sigma_1^{(2)}, \\
\bh^{(3)}=\gamma_0\sigma_0^{(3)}+\gamma_0\sigma_0^{(3)}.
\end{eqnarray*}
Then we have
\begin{eqnarray}
\tr_{x]}\left[\prod_{i\in \overrightarrow{S(x)}}K_{<x,i>}\prod_{i\in \overrightarrow{S(x)}}\bh^{(i)}\prod_{i\in \overleftarrow{S(x)}}K_{<x,i>}\right]=\Bigl\langle Ah, \sigma^{(x)}\Bigr\rangle,
\end{eqnarray}
where as before $A$ is a matrix given by \eqref{mainMatrix}, and here we assume that $\sigma^{(x)}=\left(\sigma_0^{(x)},\sigma_1^{(x)}\right),$ $h=(h_0,h_1)$ are vectors and  $\Bigl\langle \cdot, \cdot\Bigr\rangle$ stands for the standard inner product of vectors.
\end{corollary}

Let us consider the following elements:
\begin{eqnarray}
\sigma_0^{\Lambda}:=\bigotimes_{x\in \Lambda}\sigma_0^{(x)}\in\cb_{\Lambda}, \ \Lambda\subset\Lambda_n,\quad\quad
\sigma_1^{\overrightarrow{S(x)},1}:=\sigma_1^{(1)}\otimes\sigma_0^{(2)}\otimes\sigma_0^{(3)}\in \cb_{S(x)},\\
\sigma_1^{\overrightarrow{W}_{n+1},1}:=\sigma_1^{\overrightarrow{S(x_{W_n}^{(1)})},1}
\otimes\sigma_0^{\overrightarrow{W}_{n+1}\setminus\overrightarrow{S(x_{W_n}^{(1)})}}\in\cb_{W_{n+1}},\\
\label{elementa}a_{\sigma_{1}}^{\Lambda_{n+1}}:=\bigotimes_{i=0}^n\sigma_0^{\overrightarrow{W}_{i}}\otimes\sigma_1^{\overrightarrow{W}_{n+1},1}
\in\cb_{\Lambda_{n+1}}.
\end{eqnarray}

\begin{proposition}\label{evaluatestateahlpha}
Let $\ffi^{(f)}_{w_0(\a_0),{\bh(\alpha_0)}}$ be a forward quantum
$d-$Markov chain corresponding to the model
\eqref{compactformofK<u,v>} with boundary conditions
$\bh^{(x)}=\alpha_0\sigma_0^{(x)}$ for all $x\in L,$ where
$\alpha_0=\frac{1}{\cosh^3\beta}.$ Let
$a_{\sigma_{1}}^{\Lambda_{N+1}}$ be an element given by
\eqref{elementa} and  $\beta\in(\beta_{*},\beta^{*}).$ Then one has
$\ffi^{(f)}_{w_0(\a_0),{\bh(\alpha_0)}}\left(a_{\sigma_{1}}^{\Lambda_{N+1}}\right)=0,$
for any $N\in\bn.$
\end{proposition}
\begin{proof}
Due to \eqref{eq2} (see Theorem \ref{compa}) the compatibility condition holds
$\ffi^{(n+1,f)}_{w_0(\a_0),\bh(\alpha_0)}\lceil_{\cb_{\L_n}}=\ffi^{(n,f)}_{w_0(\a_0),\bh(\alpha_0)}.$ Therefore,
\begin{eqnarray}
\ffi^{(f)}_{w_0(\a_0),{\bh(\alpha_0)}}\left(a_{\sigma_{1}}^{\Lambda_{N+1}}\right)=
w-\lim_{n\to\infty}\ffi^{(n,f)}_{w_0(\a_0),\bh(\alpha_0)}
\left(a_{\sigma_{1}}^{\Lambda_{N+1}}\right)=\ffi^{(N+1,f)}_{w_0(\a_0),\bh(\alpha_0)}\left(a_{\sigma_{1}}^{\Lambda_{N+1}}\right).
\end{eqnarray}
Taking into account $w_0(\a_0)=\frac{1}{\alpha_0}\sigma_0^{(0)}$ and due to Proposition \ref{state^nwithW_n}, it is enough to evaluate the following
\begin{eqnarray}
\ffi^{(N+1,f)}_{w_0(\a_0),\bh(\alpha_0)}\left(a_{\sigma_{1}}^{\Lambda_{N+1}}\right)
&=&\tr\left(\cw_{N+1]}\left(a_{\sigma_{1}}^{\Lambda_{N+1}}\right)\right)\nonumber\\
&=&\frac{1}{\alpha_0}\tr\left[K_{[0,1]}\cdots K_{[N,N+1]}\bh_{N+1}K^{*}_{[N,N+1]}\cdots K^{*}_{[0,1]}
a_{\sigma_{1}}^{\Lambda_{N+1}}\right]\nonumber\\
&=&\frac{1}{\alpha_0}\tr\Bigl[K_{[0,1]}\cdots K_{[N-1,N]}\Bigr.\nonumber\\
&&\left.\quad\quad\quad\quad\tr_{N]}\left[K_{[N,N+1]}\bh_{N+1}K^{*}_{[N,N+1]}\sigma_1^{\overrightarrow{W}_{N+1},1}\right]K^{*}_{[N-1,N]}\cdots K^{*}_{[0,1]}\right].\nonumber
\end{eqnarray}
Now let us calculate $\widetilde{\bh}_N:=\tr_{N]}\left[K_{[N,N+1]}\bh_{N+1}K^{*}_{[N,N+1]}\sigma_1^{\overrightarrow{W}_{n+1},1}\right].$ Since $K_{<u,v>}$ is a self-adjoint, we then get
\begin{eqnarray*}
{\widetilde{\bh}}_N&=&\tr_{\left.x^{(1)}_{W_{N}}\right]}\left[\prod_{y\in \overrightarrow{S(x^{(1)}_{W_{N}})}}K_{\left\langle x^{(1)}_{W_{N}},y\right\rangle}\prod_{y\in \overrightarrow{S(x^{(1)}_{W_{N}})}}\bh^{(y)}\prod_{y\in \overleftarrow{S(x^{(1)}_{W_{N}})}}K_{\left\langle x^{(1)}_{W_{N}},y\right\rangle}\sigma_1^{\overrightarrow{S(x_{W_N}^{(1)})},1}\right]\otimes\\
&&\bigotimes_{x\in \overrightarrow{W}_{N}\setminus x^{(1)}_{W_{N}}}\tr_{x]}\left[\prod_{y\in \overrightarrow{S(x)}}K_{<x,y>}\prod_{y\in \overrightarrow{S(x)}}\bh^{(y)}\prod_{y\in \overleftarrow{S(x)}}K_{<x,y>}\right].
\end{eqnarray*}
We know that
\begin{eqnarray}\label{hhhhhh}
\tr_{x]}\left[\prod_{y\in \overrightarrow{S(x)}}K_{<x,y>}\prod_{y\in \overrightarrow{S(x)}}\bh^{(y)}\prod_{y\in \overleftarrow{S(x)}}K_{<x,y>}\right]=\bh^{(x)},
\end{eqnarray}
for every $x\in \overrightarrow{W}_{N}\setminus x^{(1)}_{W_{N}}.$ Therefore, one can
easily check that
\begin{eqnarray}
\tr_{\left.x^{(1)}_{W_{N}}\right]}\left[\prod_{y\in \overrightarrow{S(x^{(1)}_{W_{N}})}}K_{\left\langle x^{(1)}_{W_{N}},y\right\rangle}\prod_{y\in \overrightarrow{S(x^{(1)}_{W_{N}})}}\bh^{(y)}\prod_{y\in \overleftarrow{S(x^{(1)}_{W_{N}})}}K_{\left\langle x^{(1)}_{W_{N}},y\right\rangle}\sigma_1^{\overrightarrow{S(x_{W_N}^{(1)})},1}\right]={\widetilde{\bh}}^{(x^{(1)}_{W_{N}})},
\end{eqnarray}
where
\begin{eqnarray*}
{\widetilde{\bh}}^{(x^{(1)}_{W_{N}})}=\alpha_1\sigma_1^{(x^{(1)}_{W_{N}})},\quad \quad \alpha_1=\sinh\beta\cosh^5\beta.
\end{eqnarray*}
Hence, we obtain
\begin{eqnarray*}
{\widetilde{\bh}}_N={\widetilde{\bh}}^{(x^{(1)}_{W_{N}})}\bigotimes_{x\in \overrightarrow{W}_{N}\setminus x^{(1)}_{W_{N}}}\bh^{(x)}.
\end{eqnarray*} Therefore, one finds
\begin{eqnarray*}
\ffi^{(N+1,f)}_{w_0,\bh(\alpha_0)}\left(a_{\sigma_{1}}^{\Lambda_{N+1}}\right)&=&\frac{1}{\alpha_0}\tr\Bigl[K_{[0,1]}\cdots K_{[N-2,N-1]}\Bigr.\nonumber\\
&&\left.\quad\quad\quad\quad\quad\tr_{N-1]}\left[K_{[N-1,N]}{\widetilde{\bh}}_{N}K^{*}_{[N-1,N]}\right]K^{*}_{[N-2,N-1]}\cdots K^{*}_{[0,1]}\right].
\end{eqnarray*}
So, after $N$ times applying Corollary \eqref{traceforalpha_0}, we get
\begin{eqnarray*}
\ffi^{(N+1,f)}_{w_0,\bh(\alpha_0)}\left(a_{\sigma_{1}}^{\Lambda_{N+1}}\right)
=\alpha^{2N-1}_0\alpha^{N}_1(\sinh\beta\cosh^4\beta)^N\tr(\sigma_1^{(0)})=0.
\end{eqnarray*}
This completes the proof.
\end{proof}

\begin{proposition}\label{evaluatestategamma}
Let $\ffi^{(f)}_{w_0(\gamma),{\bh(\gamma)}}$ be a forward quantum $d-$Markov chain corresponding to the model \eqref{compactformofK<u,v>} with boundary conditions $\bh^{(x)}=\gamma_0\sigma_0^{(x)}+\gamma_1\sigma_1^{(x)}$ for all $x\in L,$ where $\gamma_0=\sqrt{DE}$ and $\gamma_1=\sqrt{E}.$ Let $a_{\sigma_{1}}^{\Lambda_{N+1}}$ be an element given by \eqref{elementa} and  $\beta\in(\beta_{*},\beta^{*}).$ Then one has
\begin{eqnarray}
\ffi^{(f)}_{w_0,{\bh(\gamma)}}\left(a_{\sigma_{1}}^{\Lambda_{N+1}}\right)=\frac{1}{\gamma_0}\Bigl\langle A^{N}h_{\gamma_0,\gamma_1},e\Bigr\rangle\quad \quad \forall N\in\bn,
\end{eqnarray}
where $A$ is a matrix given by \eqref{mainMatrix}, $\Bigl\langle\cdot,\cdot\Bigr\rangle$ is the standard inner product of vectors and $e=(1,0),$ $h_{\gamma_0,\gamma_1}=(h_0,h_1)$ are vectors with
\begin{eqnarray}\label{0h}
h_0&=&\gamma_0^2\gamma_1(\sinh^2\beta\cosh\beta(1+\cosh\beta)+\cosh^5\beta)+\gamma^3_1\sinh^2\beta\cosh^2\beta,\\
\label{1h}
h_1&=&\gamma_0^3\sinh\beta\cosh^5\beta+\gamma_0\gamma_1(\sinh\beta\cosh^3\beta(1+\cosh\beta)+\sinh^3\beta\cosh^2\beta).
\end{eqnarray}
\end{proposition}
\begin{proof}
Again the compatibility condition yields that
\begin{eqnarray}
\ffi^{(f)}_{w_0,{\bh(\gamma)}}\left(a_{\sigma_{1}}^{\Lambda_{N+1}}\right)=w-\lim_{n\to\infty}\ffi^{(n,f)}_{w_0,\bh(\gamma)}
\left(a_{\sigma_{1}}^{\Lambda_{N+1}}\right)=\ffi^{(N+1,f)}_{w_0,\bh(\gamma)}\left(a_{\sigma_{1}}^{\Lambda_{N+1}}\right).
\end{eqnarray}
Noting that if $\bh^{(0)}=\gamma_0\sigma_0^{(0)}+\gamma_1\sigma_1^{(0)}$ then one of the solutions of the equation $\tr(wa_0\bh^{(0)})=1$ w.r.t. $w_0$ is $w_0(\gamma)=\frac{1}{\gamma_0}\sigma_0^{(0)},$ and due to Proposition \ref{state^nwithW_n}, it is enough to evaluate the following
\begin{eqnarray}
\ffi^{(N+1,f)}_{w_0,\bh(\gamma)}\left(a_{\sigma_{1}}^{\Lambda_{N+1}}\right)
&=&\tr\left(\cw_{N+1]}\left(a_{\sigma_{1}}^{\Lambda_{N+1}}\right)\right)\nonumber\\
&=&\frac{1}{\gamma_0}\tr\left[K_{[0,1]}\cdots K_{[N,N+1]}\bh_{N+1}K^{*}_{[N,N+1]}\cdots K^{*}_{[0,1]}
a_{\sigma_{1}}^{\Lambda_{N+1}}\right]\nonumber\\
&=&\frac{1}{\gamma_0}\tr\Bigl[K_{[0,1]}\cdots K_{[N-1,N]}\Bigr.\nonumber\\
&&\left.\quad\quad\quad\tr_{N]}\left[K_{[N,N+1]}\bh_{N+1}K^{*}_{[N,N+1]}\sigma_1^{\overrightarrow{W}_{N+1},1}\right]K^{*}_{[N-1,N]}\cdots K^{*}_{[0,1]}\right].\nonumber
\end{eqnarray}
Let us calculate $\widetilde{\bh}_N:=\tr_{N]}\left[K_{[N,N+1]}\bh_{N+1}K^{*}_{[N,N+1]}\sigma_1^{\overrightarrow{W}_{n+1},1}\right].$ Self-adjointness of $K_{<u,v>}$ implies that
\begin{eqnarray*}
{\widetilde{\bh}}_N&=&\tr_{\left.x^{(1)}_{W_{N}}\right]}\left[\prod_{y\in \overrightarrow{S(x^{(1)}_{W_{N}})}}K_{\left\langle x^{(1)}_{W_{N}},y\right\rangle}\prod_{y\in \overrightarrow{S(x^{(1)}_{W_{N}})}}\bh^{(y)}\prod_{y\in \overleftarrow{S(x^{(1)}_{W_{N}})}}K_{\left\langle x^{(1)}_{W_{N}},y\right\rangle}\sigma_1^{\overrightarrow{S(x_{W_N}^{(1)})},1}\right]\otimes\\
&&\bigotimes_{x\in \overrightarrow{W}_{N}\setminus x^{(1)}_{W_{N}}}\tr_{x]}\left[\prod_{y\in \overrightarrow{S(x)}}K_{<x,y>}\prod_{y\in \overrightarrow{S(x)}}\bh^{(y)}\prod_{y\in \overleftarrow{S(x)}}K_{<x,y>}\right].
\end{eqnarray*}
It follows from \eqref{hhhhhh} that
\begin{eqnarray*}
\tr_{\left.x^{(1)}_{W_{N}}\right]}\left[\prod_{y\in \overrightarrow{S(x^{(1)}_{W_{N}})}}K_{\left\langle x^{(1)}_{W_{N}},y\right\rangle}\prod_{y\in \overrightarrow{S(x^{(1)}_{W_{N}})}}\bh^{(y)}\prod_{y\in \overleftarrow{S(x^{(1)}_{W_{N}})}}K_{\left\langle x^{(1)}_{W_{N}},y\right\rangle}\sigma_1^{\overrightarrow{S(x_{W_N}^{(1)})},1}\right]
={\widetilde{\bh}}^{(x^{(1)}_{W_{N}})},
\end{eqnarray*}
where
\begin{eqnarray*}
{\widetilde{\bh}}^{(x^{(1)}_{W_{N}})}&=&h_0\sigma_0^{(x^{(1)}_{W_{N}})}+h_1\sigma_1^{(x^{(1)}_{W_{N}})},\\
h_0&=&\gamma_0^2\gamma_1(\sinh^2\beta\cosh\beta(1+\cosh\beta)+\cosh^5\beta)+\gamma^3_1\sinh^2\beta\cosh^2\beta,\\
h_1&=&\gamma_0^3\sinh\beta\cosh^5\beta+\gamma_0\gamma_1(\sinh\beta\cosh^3\beta(1+\cosh\beta)+\sinh^3\beta\cosh^2\beta).
\end{eqnarray*}
 Thus we obtain
\begin{eqnarray*}
{\widetilde{\bh}}_N={\widetilde{\bh}}^{(x^{(1)}_{W_{N}})}\bigotimes_{x\in \overrightarrow{W}_{N}\setminus x^{(1)}_{W_{N}}}\bh^{(x)}.
\end{eqnarray*} Therefore, one gets
\begin{eqnarray*}
\ffi^{(N+1,f)}_{w_0,\bh(\gamma)}\left(a_{\sigma_{1}}^{\Lambda_{N+1}}\right)&=&\frac{1}{\gamma_0}\tr\Bigl[K_{[0,1]}\cdots K_{[N-2,N-1]}\Bigr.\nonumber\\
&&\left.\quad\quad\quad\quad\quad\tr_{N-1]}\left[K_{[N-1,N]}{\widetilde{\bh}}_{N}K^{*}_{[N-1,N]}\right]K^{*}_{[N-2,N-1]}\cdots K^{*}_{[0,1]}\right].
\end{eqnarray*}
Again applying $N$ times Corollary \eqref{traceforgamma_0gamma_1}, one finds
\begin{eqnarray*}
\ffi^{(N+1,f)}_{w_0,\bh(\gamma)}\left(a_{\sigma_{1}}^{\Lambda_{N+1}}\right)
&=&\frac{1}{\gamma_0}\tr\left[\Bigl\langle A^Nh_{\gamma_0,\gamma_1}, \sigma^{(0)}\Bigr\rangle\right]
=\frac{1}{\gamma_0}\Bigl\langle A^Nh_{\gamma_0,\gamma_1}, e\Bigr\rangle.
\end{eqnarray*}
here as before $e=(1,0),$ $h_{\gamma_0,\gamma_1}=(h_0,h_1)$ are vectors, and $A$ is a matrix given by \eqref{mainMatrix}. This completes the proof.
\end{proof}

To prove our main result we are going to use the following theorem
(see \cite{BR}, Corollary 2.6.11).
\begin{theorem}\label{br-q}
Let $\varphi_1,$ $\varphi_2$ be two states on a quasi-local algebra $\ga=\cup_{\Lambda}\ga_\Lambda$. The states $\varphi_1,$ $\varphi_2$ are  quasi-equivalent if and only if for any given $\varepsilon>0$ there exists a finite volume $\Lambda\subset L$ such that $\|\varphi_1(a)-\varphi_2(a)\|<\varepsilon \|a\|$ for all $a\in B_{\Lambda^{'}}$ with $\Lambda^{'}\cap\Lambda=\emptyset.$
\end{theorem}

Now by means of the Theorem \ref{br-q} we will show that the states
$\ffi^{(f)}_{w_0,{\bh(\alpha_0)}}$ and
$\ffi^{(f)}_{w_0,{\bh(\gamma)}}$ are not quasi-equivalent. Namely,
we have the following

\begin{theorem}
Let $\beta\in(\beta_{*},\beta^{*})$ and
$\ffi^{(f)}_{w_0(\a_0),{\bh(\alpha_0)}},$
$\ffi^{(f)}_{w_0(\gamma),{\bh(\gamma)}}$ be two forward quantum
$d-$Markov chains corresponding to the model
\eqref{compactformofK<u,v>} with two boundary conditions
$\bh^{(x)}=\alpha_0\sigma_0^{(x)},$  $\forall x\in L$ and
$\bh^{(x)}=\gamma_0\sigma_0^{(x)}+\gamma_1\sigma_1^{(x)},$  $\forall
x\in L,$ respectively, here as before
$\alpha_0=\frac{1}{\cosh^3\beta},$ $\gamma_0=\sqrt{DE},$ and
$\gamma_1=\sqrt{E}.$ Then $\ffi^{(f)}_{w_0,{\bh(\alpha_0)}}$ and
$\ffi^{(f)}_{w_0,{\bh(\gamma)}}$ are not quasi-equivalent.
\end{theorem}
\begin{proof}
Let $a_{\sigma_{1}}^{\Lambda_{N+1}}$ be an
 element given by
\eqref{elementa}. It is clear that
$\left\|a_{\sigma_{1}}^{\Lambda_{N+1}}\right\|=1,$ for all
$N\in\bn.$

If $\beta\in(\beta_{*},\beta^{*})$, then according to Propositions \ref{evaluatestateahlpha} and \ref{evaluatestategamma}, we have
\begin{eqnarray}
\label{valueofstatealpha}\ffi^{(f)}_{w_0(\a_0),{\bh(\alpha_0)}}\left(a_{\sigma_{1}}^{\Lambda_{N+1}}\right)&=&0,\\
\label{valueofstategamma0}
\ffi^{(f)}_{w_0(\gamma),{\bh(\gamma)}}\left(a_{\sigma_{1}}^{\Lambda_{N+1}}\right)&=&\frac{1}{\gamma_0}\Bigl\langle A^{N}h_{\gamma_0,\gamma_1},e\Bigr\rangle
\end{eqnarray}
for all $N\in\bn,$ here as before $e=(1,0),$ $h_{\gamma_0,\gamma_1}=(h_0,h_1)$ (see \eqref{0h},\eqref{1h})  and $A$ is given by \eqref{mainMatrix}.
Then from \eqref{valueofstategamma0} with Proposition \ref{A-N} one finds
\begin{eqnarray}\label{valueofstategamma}
\ffi^{(f)}_{w_0(\gamma),{\bh(\gamma)}}\left(a_{\sigma_{1}}^{\Lambda_{N+1}}\right)=\frac{x_1^2h_1+x_1y_1\sinh\beta h_2}{\gamma_0(x_1^2+y_1^2\sinh\beta)}+\frac{y_1^2\sinh\beta h_1-x_1y_1\sinh\beta h_2}{\gamma_0(x_1^2+y_1^2\sinh\beta)}\lambda_2^N,
\end{eqnarray}
where $\lambda_2$ is an eigenvalue of $A$ and $(x_1,y_1)$ is an
eigenvector of the matrix $A$ corresponding to the eigenvalue
$\lambda_1=1$ (see Proposition \ref{A-N}). Due to Propositions
\ref{A-m}(iv) and \ref{A-N} one has $0<\lambda_2<1,$ which implies
the existence  $N_0\in \bn$ such that
\begin{eqnarray}\label{estimateofvarepsilon}
\left|\frac{x_1^2h_1+x_1y_1\sinh\beta h_2}{\gamma_0(x_1^2+y_1^2\sinh\beta)}+\frac{y_1^2\sinh\beta h_1-x_1y_1\sinh\beta h_2}{\gamma_0(x_1^2+y_1^2\sinh\beta)}\lambda_2^N\right|\ge \frac{x_1^2h_1+x_1y_1\sinh\beta h_2}{2\gamma_0(x_1^2+y_1^2\sinh\beta)}
\end{eqnarray}
for all $N>N_0.$

Now putting $\varepsilon_0=\dfrac{x_1^2h_1+x_1y_1\sinh\beta h_2}{2\gamma_0(x_1^2+y_1^2\sinh\beta)}$ and using \eqref{valueofstatealpha}, \eqref{valueofstategamma}, \eqref{estimateofvarepsilon} we obtain
\begin{eqnarray*}
\left|\ffi^{(f)}_{w_0(\a_0),{\bh(\alpha_0)}}\left(a_{\sigma_{1}}^{\Lambda_{N+1}}\right)
-\ffi^{(f)}_{w_0(\gamma),{\bh(\gamma)}}\left(a_{\sigma_{1}}^{\Lambda_{N+1}}\right)\right|\ge \varepsilon_0\left\|a_{\sigma_{1}}^{\Lambda_{N+1}}\right\|,
\end{eqnarray*}
for all $N>N_0,$ which means $\ffi^{(f)}_{w_0(\a_0),{\bh(\alpha_0)}}$ and $\ffi^{(f)}_{w_0(\gamma),{\bh(\gamma)}}$ are not quasi-equivalent. This completes the proof.
\end{proof}

From the proved theorem we immediately get the occurrence of the phase transition for the model \eqref{compactformofK<u,v>} on the Cayley tree of order 3 in the regime $\beta\in(\beta_{*},\beta^{*})$. This completely proves our main Theorem \ref{main}.

\section{Some observation}

In this section we define a continuous function, depending on the model, such that its first order derivative has
discontinuity at the critical values of the phase phase transition.

First denote
\begin{equation}\label{freeE1}
\widetilde{K}_n(\beta)=K_{[0,1]}K_{[1,2]}\cdots K_{[n+1,n]}\mathbf{w}_{|W_{n+1}|}^{1/2},
\end{equation}
where
$$\mathbf{w}_{|W_{n+1}|}^{1/2}:=\bigotimes_{x\in \overrightarrow{W}_{n+1}}w^{1/2}(\beta).$$

Define a function $F: \br_{+}\to\br$ by the following formula
\begin{eqnarray}\label{freeE1}
\b F(\b)=\lim_{n\to\infty}\frac{1}{|V_n|}\log \tr \left(\widetilde{K}_n(\beta)\widetilde{K}^{*}_n({\beta})\right).
\end{eqnarray}

In what follows, we will consider the function $F(\beta)$ given by \eqref{freeE1} corresponding to the model \eqref{compactformofK<u,v>} with mixed boundary conditions $\omega(\alpha_0)=\frac{1}{\alpha_0}\sigma_0$, i.e.  $\bh^{(x)}=\alpha_0\sigma_0^{(x)},$  $\forall x\in L$ for $\beta\in (0,\beta_{*}]\cup[\beta^{*},\infty)$ and $\omega(\gamma_0)=\frac{1}{\gamma_0}\sigma_0,$ $\bh^{(x)}=\gamma_0\sigma_0^{(x)}+\gamma_1\sigma_1^{(x)},$  $\forall x\in L$ for $\beta\in (\beta_{*},\beta^{*}),$ here as before $\alpha_0=\frac{1}{\cosh^3\beta},$ $\gamma_0=\sqrt{DE},$ and $\gamma_1=\sqrt{E}.$

We have the following result.

\begin{theorem}
Let $F: \br_{+}\to\br$ be a function given by \eqref{freeE1}. Then the following assertion holds to be true:
\begin{itemize}
  \item [(i)] $F(\beta)$ is a continuous function on $\br_{+}$;
  \item [(ii)]The derivative function $F'(\beta)$ has the first order discontinuity at the points $\beta_{*}$ and $\beta^{*}.$
\end{itemize}
\end{theorem}

\begin{proof} Let us evaluate the value of the function $F(\beta)$ on the ranges $\beta\in (0,\beta_{*}]\cup[\beta^{*},\infty)$ and  $\beta\in (\beta_{*},\beta^{*})$, respectively.

Now assume that  $\beta\in (0,\beta_{*}]\cup[\beta^{*},\infty)$, then using the same argument as in \eqref{uniq} one gets
\begin{eqnarray}\label{freeE2}
\tr \left(\widetilde{K}_n(\beta)\widetilde{K}^{*}_n({\beta})\right)=\frac{1}{\alpha_0^{|W_{n+1}|}}\cdot \frac{\alpha_0}{\alpha_0^{|W_{n+1}|}}
\ffi^{(n,f)}_{w_0(\alpha_0),\bh(\alpha_0)}(\id)=\frac{\alpha_0}{\alpha_0^{2|W_{n+1}|}}
\end{eqnarray}
Hence, taking into account
$\lim\limits_{n\to\infty}\frac{|W_{n+1}|}{|V_n|}=2$ with \eqref{freeE2},\eqref{freeE1} we obtain
\begin{eqnarray*}
\beta F(\b)=-4\log\alpha_0(\beta),
\end{eqnarray*}
for all $\beta\in (0,\beta_{*}]\cup[\beta^{*},\infty).$

Let $\beta\in (\beta_{*},\beta^{*})$. Then in this setting, similarly as above, one derives
\begin{eqnarray}\label{freeE3}
\tr \left(\widetilde{K}_n(\beta)\widetilde{K}^{*}_n({\beta})\right)=\frac{1}{\gamma_0^{|W_{n+1}|}}\cdot \frac{\alpha_0}{\alpha_0^{|W_{n+1}|}}
\ffi^{(n,f)}_{w_0(\alpha_0),\bh(\alpha_0)}(\id)=\frac{\alpha_0}{\alpha_0^{|W_{n+1}|}\gamma_0^{|W_{n+1}|}}
\end{eqnarray}
Therefore,
\begin{eqnarray*}
\beta F(\b)=-2\log\big(\alpha_0(\beta)\gamma_0(\beta)\big),
\end{eqnarray*}
for all $\beta\in (\beta_{*},\beta^{*}).$
Thus, we have
\begin{eqnarray}
\beta F(\beta)=\left\{\begin{array}{l}
                        -4\log(\alpha_0(\beta)), \quad \ \  \beta\in (0,\beta_{*}]\cup[\beta^{*},\infty)\\
                        -2\log\big(\alpha_0(\beta)\gamma_0(\beta)\big), \quad  \beta\in (\beta_{*},\beta^{*})
                      \end{array}
\right.
\end{eqnarray}

Using \eqref{A1B1}-\eqref{D} one can calculate that
\begin{eqnarray*}
\gamma_0(\beta)=\sqrt{D(\b)E(\b)}=\sqrt{\frac{1}{B_2(\b)}+\frac{A_2(\b)(B_2(\b)-B_1(\b))}{B_2(\b)(A_2(\b)B_1(\b)-B_2(\b)A_1(\b))}}.
\end{eqnarray*}

Due to $B_2(\b_{*})=B_1(\b_{*})$, $B_2(\b^{*})=B_1(\b^{*})$ we have
\begin{eqnarray*}
\lim\limits_{\beta\to\beta_{*}+0}\gamma_0(\beta)=\alpha_0(\beta_{*}),\quad \quad \lim\limits_{\beta\to\beta^{*}-0}\gamma_0(\beta)=\alpha_0(\beta^{*}).
\end{eqnarray*}
This means that $F(\beta)$ is a continuous function on $(0,\infty).$

It is clear that $\alpha_0(\beta)$ and $\gamma_0(\beta)$ are differentiable functions on $(0,\beta_{*}]\cup[\beta^{*},\infty)$ and $(\beta_{*},\beta^{*})$ respectively.

One can easily check that
\begin{eqnarray*}
F'(\beta)\mid_{\b=\b_{*}+0}-F'(\beta)\mid_{\b=\b_{*}-0}&=&
\frac{A_2(\b_{*})(B_1'(\b_{*})-B_2'(\b_{*}))}{(A_2(\b_{*})-A_1(\b_{*}))B_2(\b_{*})\b_{*}}\neq 0,\\
F'(\beta)\mid_{\b=\b^{*}-0}-F'(\beta)\mid_{\b=\b^{*}+0}&=&
 \frac{A_2(\b^{*})(B_1'(\b^{*})-B_2'(\b^{*}))}{(A_2(\b^{*})-A_1(\b^{*}))B_2(\b^{*})\b^{*}}\neq 0,
\end{eqnarray*}
which shows that the derivative function $F'(\beta)$ has the first order discontinuity at the points $\beta_{*}$ and $\beta^{*}.$
\end{proof}

\section{Conclusions}

It is know (see \cite{BR}) if a tree is not one-dimensional lattice,
then it is expected (from a physical point of view) the existence of
a phase transition for quantum Markov chains constructed over such a
tree. In this paper, using a tree structure of graphs, we gave a
construction of quantum Markov chains on a Cayley tree, which
generalizes the construction of \cite{Ac87} to trees.  By means of
such constructions, we have established the existence of a phase
transition for quantum Markov chains associated with $XY$-model on a
Cayley tree of order three.
 By the phase transition
we means the existence of two distinct QMC for the given family of
interaction operators $\{K_{<x,y>}\}$. Note that in \cite{AMSa} we
 established the uniqueness of for the same model on the Cayley
tree of order two. Hence, results of the present paper totaly differ
from \cite{AMSa}, and show by increasing the dimension of the tree
we are getting the phase transition. In the last section we defined
a thermodynamic function, and proved that such a function is
continuous and has discontinuity at the critical values of $\b$.

%\section*{  Bibliography }

\section*{Acknowledgement} The present study have been done within
the grant FRGS0308-91 of Malaysian Ministry of Higher Education. The authors also acknowledge
the MOSTI grant 01-01-08-SF0079. This
work was done while the first named author (F.M.) was visiting the Abdus Salam International Centre for Theoretical Physics, Trieste,
Italy as a Junior Associate. He would like to thank the Centre for hospitality and financial support.

%\appendix

\section{Appendix. Proof of Lemma \ref{inequalities}}\label{appx}

$(i)$ Let $P_9(t)=t^9-t^8-t^7-t^6+2t^4+2t^3-t-1.$ One can check that
\[P_9(t)=(t-1)(t^8-t^6-2t^5-2t^4+2t^2+2t+1)\]
and $t=1$ is a root of the polynomial $P_9(t).$ It is easy to see
that $P_9(1.05)>0,$ $P_9(1.1)<0,$ $P_9(1.5)<0,$ $P_9(1.6)>0.$ This
means $P_9(t)$  has two roots $t_{*}$ and $t^{*}$ such that
$1.05<t_{*}<1.1$ and $1.5<t_{*}<1.6.$ On the other hand, due to
Descartes theorem, the number of positive roots of $P_9(t)$ is at
most the number of exchanging signs of its coefficients (i.e.
$1,-1,-1,-1,2,2,-1,-1.$) So, $P_9(t)$ has exactly three roots
$1,t_{*},t^{*}.$ It is evident  that if $t\in (1,t_{*})\cup
(t^{*},\infty)$ then $P_9(t)>0$ and $t\in (t_{*},t^{*})$ then
$P_9(t)<0.$

$(ii)$
Since $\beta > 0$  and  $\cosh\beta >\sinh\beta>0,$ we get
$$A_2-A_1=\sinh^2\beta\cosh^2\beta(2\cosh^2\beta+\cosh\beta-\sinh\beta)> 0.$$

$(iii)$ Let us denote by $t=\cosh\beta$ and  $\beta_{*}=\cosh^{-1}t_{*},$ $\beta^{*}=\cosh^{-1}t^{*}.$ One can check that
\[B_2-B_1\ge0 \quad \quad \Leftrightarrow\quad \quad P_9(t)\ge0,\]
and
\[B_2-B_1<0 \quad \quad \Leftrightarrow\quad \quad P_9(t)<0.\]

So, from $(i)$ it follows that if $\beta\in(0,\beta_{*}]\cup[\beta^{*},\infty)$ then $B_1\le B_2$ and  if $\beta\in(\beta_{*},\beta^{*})$ then $B_1> B_2.$

$(iv)$ Let us denote by $t=\cosh\beta,$ and $$Q_{10}(t)=t^{10}+4t^9+5t^8-4t^7-14t^6-6t^5+11t^4+8t^3-3t^2-2t+1.$$ One can see that
\[A_2+B_2>A_1+B_1\quad \quad \Leftrightarrow \quad\quad Q_{10}(t)>0.\] It is clear that if $\beta> 0$ then $t>1.$ One can easily get  that if $t>1$ then
$$Q_{10}(t)=t(t-1)\left((t-1)(t^7+6t^6+16t^5+22t^4+11t^3+3t(t^2-1))+2(t+1)\right)+1>0.$$

$(v)$ If $\beta\in(\beta_{*},\beta^{*})$ then $B_1-B_2>0.$ From $(iv)$ it follows that $A_2-A_1>B_1-B_2.$ This means that $D>1.$

$(vi)$ Since $1+\cosh\beta+\cosh^2\beta>1+2\cosh\beta$ and $\cosh\beta>\sinh\beta>0$ we get
\[B_1B_2-A_1A_2=\sinh\beta\cosh^3\beta\left(\cosh^5\beta(1+\cosh\beta+\cosh^2\beta)-\sinh^4\beta(1+2\cosh\beta)\right)>0.\]

It is easy to see that
\[A_2B_1-A_1B_2=\sinh^3\beta\cosh^4\beta(1+3\cosh\beta+3\cosh^2\beta+\cosh^3\beta)>0\]

$(vii)$
Let
\[Q_7(t)=t^7+2t^6-3t^4-2t^3+t^2+3t+1.\]
Then, one can easily check that
\begin{eqnarray*}
A_1A_2+3A_1B_2-A_2B_1+B_1B_2&=&\sinh\beta\cosh^3\beta Q_7(\cosh\beta).
\end{eqnarray*}
If $\beta\in (\beta_{*},\beta^{*})$ then $t\in(t_{*},t^{*})$ and
\[Q_7(t)=t(t-1)(t^5+3t^4+2t(t^2-1)+t^3-1)+2t+1>0.\]
here $t_{*}>1.$

Let \[Q_4(t)=-t^4-t^3+t^2+5t+2.\]
Then, we get
\[A_2B_1-3A_1B_2-2A_1A_2=\sinh^3\beta\cosh^3\beta Q_4(\cosh\beta).\]

One can check that $Q_4(1.7)>0$ and $Q_4(1.8)<0.$ Due to Descartes Theorem we conclude that $Q_4(t)$ has a unique positive root $\hat{t}$ such that $1.7<\hat{t}<1.8.$

If $\beta\in (\beta_{*},\beta^{*})$ then $t\in(t_{*},t^{*})$ and $t^{*}<1.7<\hat{t}.$ Then, for any $t\in(t_{*},t^{*})$ we have
$Q_4(t)>0.$

$(viii)$ It is clear that, if $\beta >0,$ then
$$\sinh\beta\cosh\beta(1+\cosh\beta)>0.$$
Now we are going to show that
\begin{eqnarray}\label{rightineq}
\sinh\beta(1+\cosh\beta)<\cosh^3\beta.
\end{eqnarray}
Noting
$$\sinh\beta=\frac{e^{\beta}-e^{-\beta}}{2},\ \ \ \cosh\beta=\frac{e^{\beta}+e^{-\beta}}{2}.$$
and letting $t=e^{\beta},$  we reduce inequality
\eqref{rightineq} to
\begin{eqnarray}\label{powersix}
t^6-2t^5-t^4+7t^2+2t+1 &>& 0
\end{eqnarray}
Since $\beta >0,$ then $t>1$. Therefore, we shall show that
\eqref{powersix} is satisfied whenever $t>1$. Now consider several
cases with respect to $t$.

{\sc Case I.} Let $t\ge1+\sqrt{2}.$ Then we have
\begin{eqnarray*}
t^6-2t^5-t^4+7t^2+2t+1 =
t^4\big(t-(1+\sqrt{2})\big)\big(t-(1-\sqrt{2})\big)+7t^2+2t+1> 0
\end{eqnarray*}

{\sc Case II.} Let $2\le t\le 1+\sqrt{2}.$ Then it is clear that
$t<\sqrt{7}.$ Therefore,
\begin{eqnarray*}
t^6-2t^5-t^4+7t^2+2t+1 = t^5(t-2)+t^2(7-t^2)+2t+1> 0
\end{eqnarray*}

{\sc Case III.} Let $\sqrt{\frac72}\le t\le 2.$ Then one gets
\begin{eqnarray*}
2(t^6-2t^5-t^4+7t^2+2t+1)& = &2t^4\bigg(t^2-\frac72\bigg)+\frac52 t^4(2-t)\\
&&+\frac32 t^2(8-t^3)+2t^2+4t+2> 0
\end{eqnarray*}

{\sc Case IV.} Let $1<t\le\sqrt{\frac72}.$ Then we have
\begin{eqnarray*}
t^6-2t^5-t^4+7t^2+2t+1 = t^4(t-1)^2+t^2(7-2t^2)+2t+1> 0
\end{eqnarray*}

Hence, the inequality \eqref{rightineq} is satisfied for all
$\beta>0.$

\end{document}